\documentclass[12pt]{article}

\usepackage{amssymb}
\usepackage{mathrsfs}
\usepackage{amsmath,amsthm}
\usepackage{array,longtable}
\usepackage[dvips]{graphicx}
\usepackage[dvips]{color}
\usepackage{colortbl}
\usepackage{wrapfig}
\usepackage{bbm}
\newcommand{\version}{October 26, 2007}

\newtheorem{thm}{Theorem}[section]
\newtheorem{lemma}{Lemma}[section]
\newtheorem{cor}{Corollary}[section]

\newcommand{\Ds}{\mathscr{D}}

\newcommand{\Rh}{\mathbb R}

\newcommand{\Ch}{\mathbb C}
\newcommand{\Zh}{\mathbb Z}

\newcommand{\Lc}{\mathcal{L}}

\newcommand{\pd}{\partial}

\newcommand{\vol}{\mathop{\mathrm{vol}}\nolimits}

\newcommand{\comment}[1]
{
\vspace{2mm}
\noindent\textit{Comment. }#1
\vspace{3mm}
}

\def\beq{\begin{equation}}                     %
\def\eeq{\end{equation}}                       %

\begin{document}
\setcounter{page}{0}
\title{
\begin{flushright}
{\small Imperial/TP/07/DMB/01}
\\
 \vspace{-5mm}{\small SPhT-T07/134}
\\
 \vspace{-6mm} {\small hep-th/yymmnnn}
\end{flushright}
\vspace{1cm} \textbf{T-duality, Gerbes and Loop Spaces}}
\author{
{Dmitriy M.~Belov$\,{}^a$,\;  Chris~M. Hull$\,{}^{a,b}$ and \;Ruben Minasian$\,{}^c$} \vspace{6mm}
\\
\emph{\normalsize ${}^a$ The Blackett Laboratory, Imperial College London}
\vspace{-2mm}
\\
\emph{\normalsize Prince Consort Road, London SW7 2AZ, UK}
\vspace{4mm}
\\
\emph{\normalsize ${}^b$ The Institute for Mathematical Sciences, Imperial College London}
\vspace{-2mm}
\\
\emph{\normalsize 48 Prince's Gardens, London SW7 2PE, U.K.}
\vspace{4mm}
\\
\emph{\normalsize ${}^c$ Service de Physique Th\'eorique,
CEA/Saclay  }
\vspace{-2mm}
\\
\emph{\normalsize 91191 Gif-sur-Yvette Cedex, France }
}

\date {~}

\maketitle
\thispagestyle{empty}

\begin{abstract} We revisit sigma models on target spaces given by a
principal torus fibration $X\to M$, and show how treating the $2$-form $B$
as a gerbe connection captures the gauging obstructions and the global
constraints on the T-duality. We show that a gerbe connection on $X$,
which is invariant with respect to the torus action, yields an affine
double torus fibration $Y$ over the base space $M$ --- the generalization
of the correspondence space. We construct a symplectic form on the
cotangent bundle to the loop space $LY$ and study the relation of its
symmetries to T-duality. We find that geometric T-duality is possible
if and only if the torus symmetry is generated by Hamiltonian vector
fields. Put differently, the obstruction to T-duality is the
non-Hamiltonian action of the symmetry group.
\end{abstract}

\vspace{1cm} $~~$ \version

\clearpage

\tableofcontents
%\clearpage

\section{Introduction}

$T$-duality is a perturbative symmetry of string theory and has played an important role in a wide number of applications, ranging from the study of WZW models to flux compactifications.
One curious aspect  of this duality is that due to mixing of the metric and the $B$ field under its action, it may connect   backgrounds which are very different   not only geometrically but also topologically.

Conventional $T$-duality is defined for certain backgrounds with
isometries
in which there is an action of an $n$-dimensional torus on the  target space $X$ preserving the metric and the
  $3$-form $H$,  which is locally given by the exterior derivative of the $B$-field.
  The equations of motion of the two-dimensional field theory then have global symmetries.
   The standard procedure for deriving the duality  starts by  gauging these  isometries, i.e. by making the symmetries local by coupling to world-sheet gauge fields. Then the target space of the model is enlarged to a space $Y$ by adding
  fields which provide the $n$ extra coordinates of $Y$ and which couple as
  Lagrange multipliers.
  One can then either eliminate the extra coordinates to recover the original model with target $X$, or integrate out the fibre coordinates of the torus fibration of $X$ to obtain the $T$-dual theory, and the
two      $T$-dual models  define the same quantum theory. Performing  the calculations classically (locally) give the well-known formulae for the transformation of the target space metric and $B$-field. However some steps of the procedure outlined above can have global obstructions, which we revisit in this paper. See
\cite{Hull:1989jk, Hull:1990ms} for the obstructions to gauging sigma models with WZ term,   e.g. \cite{Giveon:1991jj, Rocek:1991ps, Alvarez:1993qi}
 for $T$-duality from gauging sigma-models,  \cite{Hull:2006qs}
for the obstructions to $T$-duality.

\paragraph{Obstructions to gauging and $T$-duality.}
%\hspace{-.6cm}
A well studied case of   geometric $T$-duality is that of a  target space $X$ that is a principal circle fibration over a base manifold $M$. The enlarged space of the sigma model --- the correspondence space $Y$ --- can be thought of as arising from the geometrization of the $B$-field and can be represented as two independent circle fibrations over $M$ (see \cite{Bouwknegt:2003vb, Bouwknegt:2003zg, Bouwknegt:2003wp} for detailed discussions and generalizations).
Two independent projections   give the two dual geometries  - the original $X$ or the dual $\tilde X$. The two manifolds are in general topologically distinct, since in the process of dualization we are exchanging 
the curvature of the original $S^1$ bundle with
the integral of the $3$-from $H$ along the circle fibre.
This picture can be extended to a higher dimensional case of principal torus fibre $\mathbb{T}^n$, provided  the right constraints on the $B$-field are imposed.

Consider a fibration $\pi : X \to M$ with   fibre $\mathbb{T}^n$ and   connection $\Theta_I$ ($I=1,...,n$) given by  a globally well defined smooth $1$-form on $X$ with values
in $\mathfrak{t}:=\mathrm{Lie}\,\mathbb{T}^n\cong \Rh^n$. A closed $3$-form $H$ that is invariant with respect to the torus action can be written globally as
\begin{equation}
\label{iso-H-intro}
H = \pi^*  H_3 +  \pi^* H^I_2 \wedge \Theta_I  + \frac{1}{2}
 \pi^* H^{IJ}_1\wedge \Theta_I \wedge \Theta_J  + \frac{1}{6}  \pi^* H^{IJK}_0
\Theta_I \wedge \Theta_J \wedge \Theta_K  \, ,
\end{equation}
where $H_j \in \Omega^j(M; \Lambda^{3-j}\mathfrak{t})$ for $j=0,1,2,3$.\footnote{Our choice for the position of the indices $I$ ($I=1,...,n$) is somewhat unconventional - they are placed down on objects that take values in $\Lambda^{\bullet} \mathfrak{t}$ and up - on ojects in   $\Lambda^{\bullet} \mathfrak{t}^*$. This way the formulae appear to be less cluttered with different indices. Whenever
this will not cause a confusion,  the index $I$ will be suppressed.}
Gauging the sigma model  is not obstructed provided \cite{Hull:1989jk}
\begin{equation}
\label{equiv}
 \imath(K^I)\,H  = d v^I \, , \qquad  \imath(K^I)\, v^J + \imath(K^J)\, v^I = 0 \, ,
\end{equation}
where the vector fields $K^I$ ($I,J = 1, ..., n$) are the torus generators on $X$ (the Lie derivative of $\Theta$ with respect to $K^I$ vanishes), $v^I$ are globally well-defined one-forms and $ \imath(K^I)$ denotes the contraction with  the vector $K^I$. These conditions can be obtained by demanding that the gauged sigma-model action involves globally defined forms \cite{Hull:1989jk, Hull:1990ms}  and are equivalent to requiring that  the $3$-form $H$ has an equivariant extension $\bar{H}$ \cite{Witten:1991mm, FigueroaO'Farrill:1994dj}. (This means $D \bar{H} =0,$ where $D= d + \phi_I\, \imath(K^I)$ and $\phi_I$ are   two-form generators of  $\mathrm{Lie}\,\mathbb{T}^n$.   Imposing  $\bar{H}|_{\phi=0} = H$, allows us to write $\bar{H} = H - \phi_I v^I$  iff (\ref{equiv}) holds.)

 More recently   it was shown that  when the gauging and the addition of Lagrange multipliers is done together rather than in steps, these conditions can be weakened significantly  \cite{Hull:2006qs, Hull:2006va}, to become
\begin{equation}
  H^{IJK}_0 = 0 \, , \qquad H^{IJ}_1 = d B_0^{IJ} \, ,
\end{equation}
where $B_0^{IJ}$ is globally defined. The basic picture of the correspondence space still holds --- the geometrization of the $B$-field can still lead to a correspondence space $Y$ with a double-torus fibration over the base $M$. It becomes important to identify the correct connection that upon $T$-duality gets exchanged with the connection on the torus bundle. Its curvature is in the same de Rham cohomology class as the $2$-form obtained from $H$ by a single contraction with a torus generator, $H^I_2=\imath({K_I})\,H$ (and this class is no longer required to be trivial!).

One can see that, in the absence of  $B_0^{IJ}$, the $2$-form $ H^I_2 $ is closed and can be thought of as the curvature of a connection $\tilde \Theta$ on the dual principal torus fibre over $M$, ${\tilde \pi}: {\tilde X} \to M$, and we may indeed pass to the sigma model on the extended target space given by the fibrewise product $Y= X \times_M {\tilde X}$.  $T$-duality acts to  interchange $\Theta$ and $\tilde \Theta$.

\paragraph{The generalized correspondence space.}
%\hspace{-.6cm}
Much of our understanding of $T$-duality is based on the relation with gauged sigma-models, so it is interesting to investigate further the cases in which gauging is not possible.
We will focus here on one of the simplest obstructed cases, that in which $B_0^{IJ}$ is not globally defined but $B$ is invariant under the torus action. While being the simplest obstructed case, it is sufficiently nontrivial to illustrate some of the problems one encounters in attempting to perform an obstructed $T$-duality.
To discuss $T$-duality in the more general case one has to specify how the torus group
acts on the gerbe connection (the $B$-field), see \cite{Hull:2006qs} and
appendix \ref{app:t-gerbe} of this paper for more details.
The topological aspects of $T$-duality with nontrivial
$B$-fields have been discussed in
\cite{Mathai:all,Mathai:2004qc,Mathai:2004qq,Mathai:2005fd}.

As when discussing the global aspects of WZ models \cite{Felder:1988sd, Gawedzki:1987ak}, it will be crucial for our discussion to treat the $2$-form as a gerbe connection.
As before, its geometrization leads to a new enlarged space $Y$. As we shall see, upon imposing certain conditions on the gerbe structure, $Y$ has two different descriptions. It can either be viewed  as a principal torus fibration over the original torus fibration $X$ with a well defined connection form $\Theta_{\#}$, or  as an affine $2n$-dimensional torus fibration over the base $M$.  As we shall see the affine connection $( \Theta, \tilde \Theta)$ has the following gluing conditions on twofold overlaps $M_{\alpha \beta}$:
\begin{equation}
\begin{pmatrix}
\tilde{\Theta}_{\alpha}
\\
\Theta_{\alpha}
\end{pmatrix}
=
\begin{pmatrix}
\mathbbmss{1} & m_{\alpha\beta}
\\
0 & \mathbbmss{1}
\end{pmatrix}
\begin{pmatrix}
\tilde{\Theta}_{\beta}
\\
\Theta_{\beta}
\end{pmatrix} \,
\end{equation}
where $m^{IJ}_{\alpha\beta}$ are skewsymmetric integral valued matrices satisfying cocycle conditions on triple overlaps (see subsection \ref{sec:gerbe} for details).  They parameterize  the non-triviality of the $B$-field. When  $m^{IJ}_{\alpha\beta}$ can be set to zero, $B_0^{IJ}$ is a globally defined smooth function (and the $T$-duality is geometric).

The principal difference between the two connections $ \Theta_{\#}$, $ \tilde \Theta$ becomes clear when
the lifting of the original $\mathbb{T}^n$ action to $Y$ is considered ---
even for well-defined $B_0^{IJ}$ . When using the connection
$\Theta_{\#}$, the torus group in general lifts only to the universal
covering group $\Rh^n$ (which will be the case even if $B_0^{IJ}$ is
constant provided the matrix $B_0^{IJ}$ has irrational values at some points on $M$, which will necessarily be the case if it is a non-constant function), while for
$\tilde{\Theta}$ it always lifts to $\mathbb{T}^n$.  As mentioned, when
$B_0^{IJ}$ is well defined, the two connections have curvatures which are
in the same de Rham cohomology class. When it is not, this discrepancy
becomes one of the principal difficulties. In this situation it is
impossible to perform the $T$-duality in the standard way at the level of
the sigma model.  In such cases, it has been proposed that $T$-duality is nonetheless possible and
gives a  T-fold \cite{Hull:2004in}.

\paragraph{T-duality as a symmetry of a loop space.}
%\hspace{-.6cm}
The phase space of the sigma model on the target $X$, given by the cotangent bundle to the loop space $LX$, has a natural symplectic structure with a closed $2$-form
\begin{equation}
\omega_X=\oint_{S^1}d\sigma\,\bigl[\delta p
+\imath(\pd_{\sigma}x)H\bigr] \, ,
\label{omegaX-intro}
\end{equation}
where $\delta$ is the differential on the loop space. We  shall see how this structure extends upon enlarging the target space from the principal torus fibration $X$ to the generalized correspondence space $Y$. The new symplectic form on $T^*LY$, $\omega_Y$,  has a natural $O(n,n, \mathbb{Z})$ action. In the unobstructed case this action simply leads to a new derivation of the old results. There are two different torus actions with two different symplectic reductions leading either back to the original sigma model on $X$ or the dual one on $\tilde X$ with the exchange of the first Chern classes of the torus fibrations and the fibrewise integrals of $h$ --- the characteristic class of the gerbe connection. When $B_0^{IJ}$ is not globally defined, the situation changes radically.  The natural extension of the construction for this case uses the globally well defined connection  $\Theta_{\#}$  (in a way similar to the reduction of a centrally extended current algebra on a torus fibre as explained in the Appendix \ref{app:reductionCourant}). However, as already mentioned, when using
$\Theta_{\#}$ the original  torus $\mathbb{T}^n$ acts as $\Rh^n$ in $T^*LY$ and one has to   deal with non-closed orbits in $Y$. This problem manifests itself in the fact that the action of the torus $\mathbb{T}^n$ on $\omega_Y$ is no longer hamiltonian. Not surprisingly, the obstruction is given by $H_1^{IJ}$.  As we shall see (Theorem \ref{thm:last}) there exists a way of writing the symplectic form   on $T^*LY$ using the affine connection $\tilde \Theta$ (see section~\ref{sec:loop} for details).

\vspace{0.5cm}\noindent The structure of the paper is as follows. In section \ref{sec:sigma}, we rederive the $T$-duality obstructions and describe the construction of the enlarged target space. The way the original gerbe structure defined the structure of this space is discussed there (and in Appendices  \ref{app:WZ}, \ref{app:t-gerbe}; Appendix \ref{app:u1} discusses the toy example of a $U(1)$ bundle reduction). $T$-duality in sigma models is discussed in section \ref{sec:stdTdual}.  In section \ref{sec:loop}, we discuss the construction of the phase space on the enlarged sigma model and the action of the obstructed $T$-duality (with the current algebra being discussed in Appendix \ref{app:reductionCourant}).

%%%%%%%%%%%%%%%%%%%%%%%%%%%%%%%%%%%%%%%%%%%%%%%%%%%

\section{Sigma-models on principal torus bundles}
\label{sec:sigma}\setcounter{equation}{0}

\subsection{Review of the obstructions  to T-duality}
\label{sec:sigmaT}
\paragraph{Principal torus bundle.}
Let $X$ be a principal torus bundle with   fibre $\mathbb{T}^n$:
\begin{equation*}
\mathbb{T}^n\hookrightarrow X\stackrel{\pi}{\longrightarrow}M.
\end{equation*}
\begin{wrapfigure}{l}{150pt}
\vspace{-17pt}
\includegraphics[width=130pt]{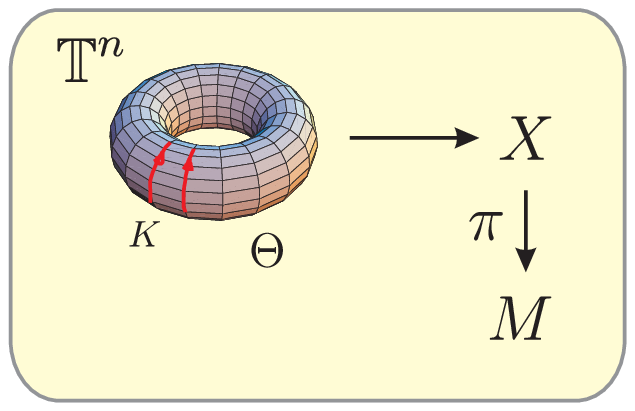}
\end{wrapfigure}
A connection on $X$ is a globally well defined smooth $1$-form $\Theta$ on $X$ with values
in $\mathfrak{t}:=\mathrm{Lie}\,\mathbb{T}^n\cong \Rh^n$.
Let $K\in \Gamma(TX\otimes \mathfrak{t}^*)$ be a fundamental vector field --- the
generator of the $\mathbb{T}^n$-action on $X$.
The connection $\Theta$ is characterized by
\begin{equation*}
\imath({K})\,\Theta=\mathbbmss{1}\in \mathfrak{t}^*\otimes \mathfrak{t}.
\end{equation*}
and the  equivariance condition
\begin{equation*}
\mathcal{L}(K)\,\Theta=0
\end{equation*}
where $\Lc(K)$ denotes the Lie derivative with respect to the vector field $K$.
These two conditions imply that $d\Theta=\pi^* F$ is a horizontal form,
$F\in\Omega^2_{\Zh}(M;\mathfrak{t})$.

It is convenient to choose a basis on $\mathfrak{t}$, so that one can
think of the connection $\Theta$ as a collection of one-forms $\{\Theta_I\}$, $I=1,\dots,n$.
We denote the corresponding fundamental vector fields by $\{\frac{\pd}{\pd\theta_I}\}$.
Note that, given a connection $\Theta$,
\begin{equation*}
\Theta_I\wedge \imath(\tfrac{\pd}{\pd\theta_I})\quad\text{and}\quad
1-\Theta_I\wedge \imath(\tfrac{\pd}{\pd\theta_I})
\end{equation*}
 are
the projection operators onto the vertical and horizontal forms respectively.
Similarly, we can decompose the differential $d$ into  a horizontal differential $\pi^*d_M$
and the vertical one $d_{\pd/\pd\theta}$:
\begin{subequations}
\begin{equation}
d=\bigl[1-\Theta_I\wedge\imath(\tfrac{\pd}{\pd\theta_I})\bigr]d
+\Theta_I\wedge\imath(\tfrac{\pd}{\pd\theta_I})\,d
=\pi^*d_M+d_{\pd/\pd\theta}
\end{equation}
where $d_M$ is the differential on $M$.
On the horizontal forms $\omega_{\texttt{hor}}$, $\imath(\frac{\pd}{\pd\theta_I})\omega_{\texttt{hor}}=0$,
one has $\imath(\frac{\pd}{\pd\theta_I})d\omega_{\texttt{hor}}=
\Lc(\frac{\pd}{\pd\theta_I})\,\omega_{\texttt{hor}}$ with $\Lc(\frac{\pd}{\pd\theta_I})\,\omega_{\texttt{hor}}$
also horizontal, and therefore
\begin{equation}
d\omega_{\texttt{hor}}=(\pi^*d_M)\,\omega_{\texttt{hor}}
+\Theta_I\wedge \Lc(\tfrac{\pd}{\pd\theta_I})\,\omega_{\texttt{hor}}.
\end{equation}
\label{hor_diff}
\end{subequations}
The lift $\pi^*d_M$ of the differential on $M$ is not nilpotent: rather
$(\pi^*d_M)^2\omega_{\texttt{hor}}=-F_I\wedge \Lc(\frac{\pd}{\pd\theta_I})\,\omega_{\texttt{hor}}$.

In the next section we will use \textit{a local description} of the torus
bundle. The following notation will be used. We choose an open cover $\{M_{\alpha}\}$
of the base $M$ by contractible open sets. We denote by $\theta_{\alpha\,I}$ ($I=1,\dots,n$),
$0\leqslant
\theta_{\alpha\,I}<1$, coordinates in the torus
fibre over the patch $M_{\alpha}$ with the gluing condition on twofold overlaps $\{M_{\alpha\beta}\}$
\begin{subequations}
\begin{equation}
\theta_{\alpha}|_{M_{\alpha\beta}}-\theta_{\beta}|_{M_{\alpha\beta}}=-\lambda_{\alpha\beta}
\label{thetagluing}
\end{equation}
where $\{\lambda_{\alpha\beta}\}$ are functions on twofold overlaps with values
in $\mathfrak{t}$ satisfying the cocycle condition on threefold overlaps: $\lambda_{\alpha\beta}
+\lambda_{\beta\gamma}+\lambda_{\gamma\alpha}=0$.
Then locally the connection $\Theta$ can be written as
\begin{equation}
\Theta|_{M_{\alpha}}=d\theta_{\alpha}+\pi^*A_{\alpha}
\quad\text{and}\quad
A_{\alpha}\bigr|_{M_{\alpha\beta}}-A_{\beta}\bigr|_{M_{\alpha\beta}}=d\lambda_{\alpha\beta}
\label{Agluing}
\end{equation}
\end{subequations}
where $A_{\alpha}$ is a $1$-form on $M_{\alpha}$ with values in $\mathfrak{t}$.

\paragraph{Restrictions on the $3$-form $H$.}
We are interested in sigma models on a target space $X$, given by a principal torus fibration,
and a Wess-Zumino term defined by a $2$-form gauge field $B$.
To be more precise, $B$ is a gerbe connection.
The implications of this description are important and will be  explained in the next subsection.
For the moment,  we are  interested in the curvature of the gerbe connection
--- a globally well defined smooth closed $3$-form $H\in\Omega^3_{\Zh}(X)$.

Since $\pi:X\to M$ is a principal torus bundle we have a free torus action on $X$.
The Wess-Zumino term is invariant with respect to this torus action (more precisely, the
holonomies\footnote{The holonomy of a $2$-form gauge field $B$
over a $2$-cycle $\Sigma$ is, roughly speaking,  $\exp (2 \pi i \int_{\Sigma} B) $ and is  defined   in a way similar to
the holonomy of a $1$-form gauge field --- see   Appendix~\ref{app:WZ} for details.
The
holonomy of a gerbe connection is an exponential of a Wess-Zumino
term: $\mathrm{Hol}(B,\Sigma)=\exp[2\pi i \mathit{WZ}(B,\Sigma)]$. } of the gerbe connection $B$ over $2$-cycles in $X$
are invariant with respect to the torus action) iff $\imath(\frac{\pd}{\pd\theta_I})H$ is an exact form.
This is a necessary condition for gauging the sigma model \cite{Hull:1989jk}.
However the conditions for $T$-duality are less restrictive: $\Lc(\frac{\pd}{\pd\theta_I})H=0$,
i.e. $\imath(\frac{\pd}{\pd\theta_I})H$ is a closed $2$-form but not necessarily an exact one \cite{Hull:2006qs}.
 Such a $3$-form $H$
can be written globally as
\begin{equation}
\label{iso-H}
H = \pi^*  H_3 + \langle \pi^* H_2, \Theta \rangle + \frac{1}{2}
\langle \pi^* H_1, \Theta \wedge \Theta \rangle + \frac{1}{6}  \langle \pi^* H_0,
\Theta \wedge \Theta \wedge \Theta \rangle \, ,
\end{equation}
where $H_j \in \Omega^j(M; \Lambda^{3-j}\mathfrak{t})$ for $j=0,1,2,3$
and $ \langle \cdot, \cdot \rangle$ denotes the natural pairing $\mathfrak{t}^*\otimes\mathfrak{t}\to\Rh$.
We use the same notation for the linear extension of this pairing to   antisymmetric
powers of $\mathfrak{t}$ and $\mathfrak{t}^*$.
For example, $\langle H_1,\Theta\wedge\Theta\rangle = H_1^{IJ}\wedge\Theta_I\wedge\Theta_J$.
The closure of $H$ implies the following equations on $\{H_j\}$
\begin{subequations}
\begin{equation}
\label{eq-H}
dH_j + \langle H_{j-1}, F \rangle =0,
\end{equation}
or using the basis we have
\begin{alignat}{2}
d_MH_3+H_2^I\wedge F_I&=0,&\qquad\qquad\qquad
d_MH_2^I+H_1^{IJ}\wedge F_J&=0,
\notag
\\
d_MH_1^{IJ}+H_0^{IJK}F_K&=0,
&\qquad
d_MH_0^{IJK}&=0.
\label{eq-Hb}
\end{alignat}
\end{subequations}

\paragraph{Double fibration.}
The contraction of the invariant $3$-form $H$ with the fundamental vector field $K$ defines
a closed $2$-form $F_{\#}\in  \Omega^2_{\Zh}(M; \mathfrak{t})$ on $X$ with integral periods
(provided that the  fundamental vector field $K$ is properly normalized, as  we will now
assume).
Using the basis in $\mathfrak{t}$ it can be written as
\begin{equation}
F_{\#}^I:=\imath(\tfrac{\pd}{\pd\theta_I})H=H_2^I-H_1^{IJ}\wedge \Theta_J
+\frac12\,H_0^{IJK}\Theta_J\wedge\Theta_K.
\label{Fhash}
\end{equation}
We would like now to geometrize this form, i.e. think of it
as a curvature of a connection $\Theta_{\#}$ on  a
principal torus bundle $\mathbb{T}^n_{\#}\hookrightarrow Y\stackrel{p}{\longrightarrow}X$
with $\mathrm{Lie}\,\mathbb{T}^n_{\#} = \mathfrak{t}^*$:\\
\vspace{-0.8cm}
\begin{wrapfigure}{l}{150pt}
\vspace{-7pt}
\includegraphics[width=130pt]{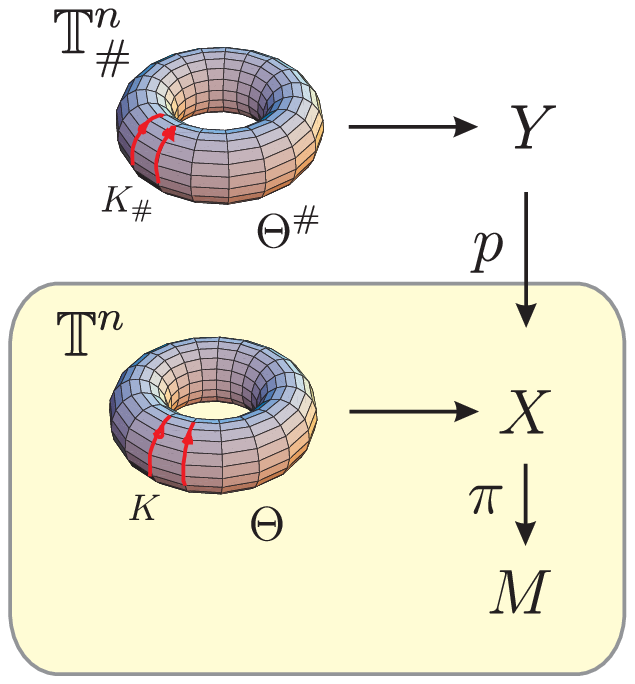}
\end{wrapfigure}
\begin{equation}
\label{theta-hash}
p^*(F_{\#}^I)= d \Theta_{\#}^I.
\end{equation}
To this end one has to construct a $2$-cocycle representing the first Chern
class of the torus bundle, such that its image in the de Rham cohomology is
$[F_{\#}]_{dR}$.

\paragraph{Torus actions on the double fibration.}
The total space $Y$ of the double fibration has a natural action of
the torus $\mathbb{T}^n_{\#}$.
It is natural to ask whether the original torus $\mathbb{T}^n$
acts on $Y$.
The connection $\Theta_{\#}$ on $Y\to X$ allows one to lift
the action of $\mathrm{Lie}\,\mathbb{T}^n$:
a fundamental vector field $K\in \Gamma(TX\otimes \mathfrak{t}^*)$ can be
lifted to $Y$ as a horizontal\footnote{This means that
$\imath({K_{\texttt{hor}}})\Theta_{\#}=0$.} vector field $K_{\texttt{hor}}\in \Gamma(TY\otimes \mathfrak{t}^*)$.
Note that these horizontal vector fields do not commute automatically.
Indeed, the commutator of two such fields $K_{\texttt{hor}}$ and $K'_{\texttt{hor}}$
is given by their contraction with the curvature
of   $\Theta_{\#}$:
\begin{equation*}
[K_{\texttt{hor}}, K'_{\texttt{hor}}] = [K,K']_{\texttt{hor}}+\imath({K_{\texttt{hor}}})
\imath({K'_{\texttt{hor}}})F_{\#} \, ,
\end{equation*}
or explicitly
\begin{equation}
[(\tfrac{\pd}{\pd\theta_I})_{\texttt{hor}},(\tfrac{\pd}{\pd\theta_J})_{\texttt{hor}}]
=-H_0^{IJK}\tfrac{\pd}{\pd\theta_{\#}^K}.
\label{horizz}
\end{equation}
Thus the vanishing of $H_0 \in \Omega^0_{\Zh}(M; \Lambda^{3}\mathfrak{t})$  is the necessary
condition for the action of $\mathrm{Lie}\,\mathbb{T}^n$ to remain abelian after lifting to $Y$.

Having lifted the action of
the Lie algebra we have not necessarily lifted the action of the Lie group as well.
To lift the torus action $\mathbb{T}^n$ to $Y$ in addition to $H_0^{IJK}=0$
we have to verify that the orbits of $(\frac{\pd}{\pd\theta_I})_{\texttt{hor}}$ are closed for
all $I=1,\dots,n$.
If this is so, then we have an action of the double torus $\mathbb{T}^n\times
\mathbb{T}^n_{\#}$ on $Y$. Otherwise, if $H_0^{IJK}=0$ but not all orbits are closed
we have the action of $\Rh^k \times \mathbb{T}^{n-k}\times\mathbb{T}^n_{\#}$ on $Y$
for some $k$ between $1$ and $n$.

A free action of $\mathbb{T}^n\times\mathbb{T}^n_{\#}$ on $Y$ means that
 $Y$ itself is a double torus fibration over
$M$. In particular,
\begin{equation}
[F_{\#}^I]_{dR}=[\pi^*H_2^I-\pi^*H_1^{IJ}\wedge\Theta_J]_{dR}
\label{Fh}
\end{equation}
must be a pullback of some de Rham cohomology class on $M$.
The first term in this expression is clearly a pullback.
However the second one is not in general a pullback
from $M$. Suppose that $H_1$ is exact, i.e. there exists a globally well defined
smooth $B_0\in \Omega^0(M;\Lambda^2 \mathfrak{t}^*)$
such that $H_1=dB_0$. Then we can rewrite \eqref{Fh} as
\begin{equation*}
[F_{\#}^I]_{dR}=\pi^*[H_2^I+B_0^{IJ}F_J]_{dR}.
\end{equation*}
Thus the necessary conditions for having a free action of the double torus on $Y$ are
$H_0^{IJK}=0$ and $H_1^{IJ}$ is exact. In this case $Y$ itself is a principal $\mathbb{T}^n
\times\mathbb{T}^n_{\#}$ bundle over $M$.
One can choose on $Y$ a connection $\tilde{\Theta}$ which respects the fact that
$Y$ is principal double torus bundle over $M$:
\begin{equation}
\tilde{\Theta}^I=\Theta_{\#}^I-B_0^{IJ}\Theta_J.
\label{tTh}
\end{equation}
It is this $\mathbb{T}^n
\times\mathbb{T}^n_{\#}$ bundle over $M$ that is referred to as the doubled torus bundle in \cite{{Hull:2004in},Hull:2006qs}.

We can summarize our discussion by the following
\begin{thm}
The contraction of the invariant $3$-form $H$ \eqref{iso-H}
with the fundamental vector field defines a closed $2$-form $F_{\#}$ with integral periods on $X$.
One can think of it as a curvature of a connection $\Theta_{\#}$ on a principal
torus bundle $\mathbb{T}^n_{\#}\hookrightarrow Y\stackrel{p}{\longrightarrow}X$: $p^*F_{\#}=d\Theta_{\#}$.
\begin{enumerate}
\item[a)] The action of $\mathrm{Lie}\,\mathbb{T}^n$, the Lie algebra of
the original torus, is abelian
on $Y$ iff $H_0^{IJK}=0$.
\item[b)] If in addition $H_1^{IJ}=dB_0^{IJ}$ is an exact form on $X$, then
$Y$ is a principal double torus bundle on $M$ with connections $\Theta$ and $\tilde{\Theta}$
defined by \eqref{tTh}.
\end{enumerate}
\end{thm}

\comment{In the next subsection we will show that if $H_1$ is \textit{not} exact
then $Y$ is an affine torus bundle over $M$ with very specific gluing functions.}

\subsection{Gerbes}
\label{sec:gerbe}
The general sigma model includes a Wess-Zumino defined by a
 $2$-form gauge field $B$. When discussing the global properties of the sigma model it is important to treat
 $B$ as a gerbe connection. In this subsection we first review the definition of a gerbe on a general manifold $X$, and then consider in detail what happens when  $X$ is a principal torus bundle and the curvature of the gerbe connection is invariant with respect to the torus action. The geometrization of the gerbe for this case will be the key to the following discussion.

The main results of this subsection are Corollary~\ref{cor:1}
and Corollary~\ref{cor:2} which
state that an invariant gerbe connection on a principal torus bundle defines:
\begin{enumerate}
\item[a)] a principal torus bundle $p:Y\to X$ with connection;
\item[b)] an affine (double) torus bundle over $M$ with an affine connection.
\end{enumerate}

\paragraph{Gerbe.} We use the formulation of a gerbe   presented in section 1.2 in
\cite{Hitchin:2005uu, Hitchin:2005in}. Choose an open
\begin{wrapfigure}{l}{145pt}
\vspace{-7pt}
\includegraphics[width=135pt]{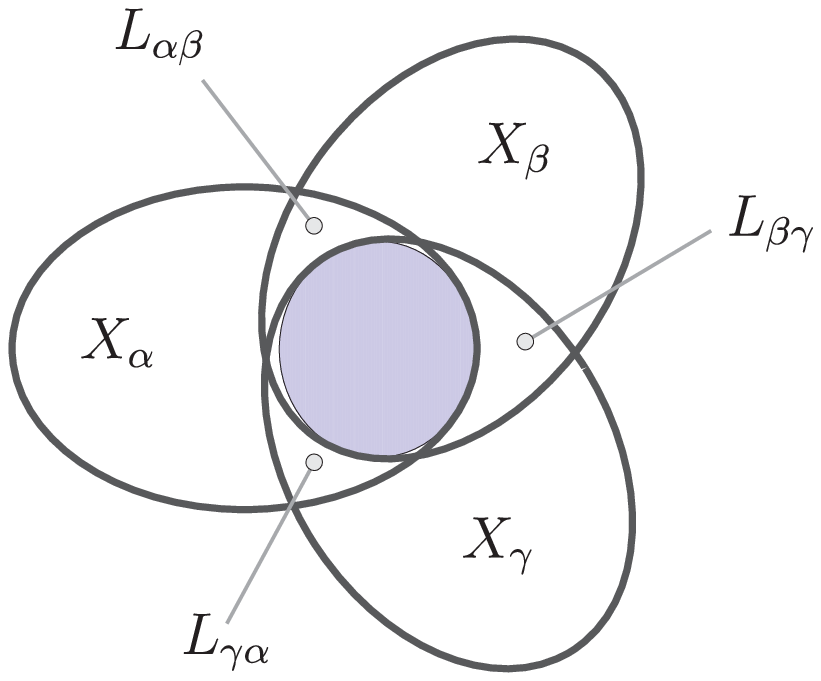}
\end{wrapfigure}
 covering $\{X_{\alpha}\}$ of $X$. Note
that these open sets do not need to be contractible. \textit{A gerbe} is defined
as the following structure: a line bundle $L_{\alpha\beta}$ on each
twofold intersection $X_{\alpha\beta}=X_{\alpha}\cap X_{\beta}$;
an isomorphism $L_{\alpha\beta}\cong L_{\beta\alpha}$;
a trivialization $f_{\alpha\beta\gamma}:X_{\alpha\beta\gamma}\to U(1)$ of the line bundle
$L_{\alpha\beta}\otimes L_{\beta\gamma}\otimes L_{\gamma\alpha}$
on each threefold intersection $X_{\alpha\beta\gamma}$; $f_{\alpha\beta\gamma}$ is a
cocycle, i.e. $\delta f_{\alpha\beta\gamma\delta}
=f_{\alpha\beta\gamma}f_{\beta\gamma\delta}^{-1}f_{\gamma\delta\alpha}
f_{\delta\alpha\beta}^{-1}=1$ on each fourfold intersection $X_{\alpha\beta\gamma\delta}$.

\textit{A gerbe with connection} is a gerbe plus a connection $A_{\alpha\beta}$
on\footnote{Connections $A_{\alpha\beta}$ on line bundles $L_{\alpha\beta}$
should not be confused with the connection $A_{\alpha}$ on a principal torus bundle
defined in subsection~\ref{sec:sigmaT}.}  the line bundle $L_{\alpha\beta}$ in each $X_{\alpha\beta}$
such that the section $f_{\alpha\beta\gamma}$ is covariantly constant with respect to
the induced connection on $L_{\alpha\beta}\otimes L_{\beta\gamma}\otimes L_{\gamma\alpha}$:
\begin{equation}
A_{\alpha\beta}+A_{\beta\gamma}+A_{\gamma\alpha}=
\frac{1}{2\pi i}f_{\alpha\beta\gamma}^{-1}\,df_{\alpha\beta\gamma} \, ,
\label{equivariance}
\end{equation}
and a two form (gerbe connection) $B_{\alpha}\in\Omega^2(X_{\alpha})$ such that
$B_{\alpha}-B_{\beta}=dA_{\alpha\beta}$ on $X_{\alpha\beta}$.

The gauge group of the gerbe is generated by a group of line bundles with connection.
Given a line bundle $L$ with connection $A$ we shift $L_{\alpha\beta}\mapsto
L|_{X_{\alpha\beta}}\otimes L_{\alpha\beta}$,
$A_{\alpha\beta}\mapsto A|_{X_{\alpha\beta}}+A_{\alpha\beta}$ and
$B_{\alpha}\mapsto B_{\alpha}+F|_{X_{\alpha}}$ where $F\in
\Omega^2_{\Zh}(X)$ is the curvature of the connection on $L$.
The gauge equivalence classes of gerbe connections form an abelian group
--- the Chiger-Simons cohomology $\check{H}^3(X)$ (for a pedagogical
introduction to Chiger-Simons cohomology see section~2 in \cite{Freed:2006yc}).

\paragraph{Gerbe on a principal torus bundle.}
As we have seen in section~\ref{sec:sigmaT}, a necessary condition
for $T$-duality is the invariance of the curvature $H$ with respect to
the torus action. Thus $H$ can be written as
\begin{equation}
H=\pi^*H_3+\langle \pi^*H_2,\Theta\rangle
+\frac12\langle \pi^*H_1,\Theta\wedge\Theta\rangle
+\frac16\langle \pi^*H_0,\Theta\wedge\Theta\wedge\Theta\rangle
\label{H}
\end{equation}
where $H_j$ for $j=0,1,2,3$ is a smooth $j$-form on $M$.
The torus bundle $X$ can be covered
by open sets $\{X_{\alpha}=\mathbb{T}^n\times M_{\alpha}\}$ where
$\{M_{\alpha}\}$ is an open covering of the base manifold $M$.

\vspace{3mm}
\noindent\textit{Structure on a coordinate patch. } In each coordinate patch $\mathbb{T}^n\times M_{\alpha}$ a gerbe connection $B_{\alpha}$
can be written as
\begin{equation}
B_{\alpha}=B_{2\alpha}+\langle B_{1\alpha},\Theta\rangle+\frac12\langle B_{0\alpha},\Theta\wedge\Theta\rangle,
\label{B}
\end{equation}
where $B_{2\alpha}$, $B_{1\alpha}$ and $B_{0\alpha}$ are horizontal
$2$-, $1$- and $0$-forms on $X_{\alpha}$.  Note that there is no $\pi^*$ in
front of these forms in the equation (\ref{B})
since a priori they can depend on the torus coordinates.
Locally the curvature $H_j|_{X_{\alpha}}$ can be written as
\begin{subequations}
\begin{align} \pi^*H_3|_{X_{\alpha}}&=(\pi^*d_M)B_{2\alpha}-\langle B_{1\alpha},\pi^*F\rangle;
\\
\pi^*H_2^{I}|_{X_{\alpha}}&=(\pi^*d_M)B_{1\alpha}^I
+\Lc(\tfrac{\pd}{\pd\theta_I})B_{2\alpha}-B_{0\alpha}^{IJ}\wedge\pi^*F_J;
\\
\pi^*H_1^{IJ}|_{X_{\alpha}}&=(\pi^*d_M)B_{0\alpha}^{IJ}
-\Lc(\tfrac{\pd}{\pd\theta_I})B_{1\alpha}^{J}
+\Lc(\tfrac{\pd}{\pd\theta_J})B_{1\alpha}^{I};
\\
\pi^*H_0^{IJK}|_{X_{\alpha}}&=\tfrac{\pd}{\pd\theta_I}B_{0\alpha}^{JK}
+\tfrac{\pd}{\pd\theta_J}B_{0\alpha}^{KI}
+\tfrac{\pd}{\pd\theta_K}B_{0\alpha}^{IJ}
\label{Hj}
\end{align}
\end{subequations}
where $\pi^*d_M$ is the horizontal exterior derivative defined in (\ref{hor_diff}).
Note that the left hand sides of the equations above do not
depend on the torus coordinates, and  thus the right hand sides should not depend on them either.

\vspace{3mm}
\noindent\textit{Structure on a twofold intersection. }
On twofold intersections $\{X_{\alpha\beta}=\mathbb{T}^n\times M_{\alpha\beta}\}$
the gerbe connections $\{B_{\alpha}\}$ are glued
by $1$-forms $\{A_{\alpha\beta}\}$ which can be written as
\begin{equation}
A_{\alpha\beta}=a_{\alpha\beta}+\langle h_{\alpha\beta},\Theta\rangle,
\label{Adecomp}
\end{equation}
where $a_{\alpha\beta}$ and $h_{\alpha\beta}$ are horizontal $1$- and $0$-forms on $X_{\alpha}$
respectively. There is no $\pi^*$ in this expression since a priori both $a_{\alpha\beta}$
and $h_{\alpha\beta}$ can depend on the torus coordinates.
The gluing condition yields
\begin{subequations}
\begin{align}
B_{2\alpha}-B_{2\beta}\bigl|_{M_{\alpha\beta}}&=(\pi^*d_M)\,a_{\alpha\beta}+\langle h_{\alpha\beta},\pi^*F\rangle,
\\
B_{1\alpha}^I-B_{1\beta}^I\bigl|_{M_{\alpha\beta}}&=(\pi^*d_M)h_{\alpha\beta}^I-\Lc(\tfrac{\pd}{\pd\theta_I})\,a_{\alpha\beta},
\\
B_{0\alpha}^{IJ}-B_{0\beta}^{IJ}\bigl|_{M_{\alpha\beta}}&=\tfrac{\pd}{\pd\theta_I}h^{J}_{\alpha\beta}
-\tfrac{\pd}{\pd\theta_J}h^{I}_{\alpha\beta}.
\end{align}
\label{BBgluing}
\end{subequations}

\vspace{3mm}
\noindent\textit{Structure on a threefold intersection. }
On threefold intersections  $\{X_{\alpha\beta\gamma}=\mathbb{T}^n\times M_{\alpha\beta\gamma}\}$
we are given   sections $f_{\alpha\beta\gamma}:X_{\alpha\beta\gamma}\to U(1)$ satisfying
the cocycle condition on fourfold intersections.
The connections $\{A_{\alpha\beta}\}$ must be such that $f_{\alpha\beta\gamma}$ is
covariantly constant \eqref{equivariance}:
\begin{equation}
a_{\alpha\beta}+a_{\beta\gamma}+a_{\gamma\alpha}\bigl|_{M_{\alpha\beta\gamma}}
=\frac{1}{2\pi i}\,(\pi^*d_M)\log f_{\alpha\beta\gamma}
\quad\text{and}\quad
h_{\alpha\beta}^I+h_{\beta\gamma}^I+h_{\gamma\alpha}^I
\bigl|_{M_{\alpha\beta\gamma}}=\frac{1}{2\pi i}\,\frac{\pd}{\pd\theta_I}\log f_{\alpha\beta\gamma}.
\label{fabc}
\end{equation}

\paragraph{T-duality constraints.}
Recall that the contraction of the fundamental vector field $\frac{\pd}{\pd\theta_I}$
with the form $H$ yields a closed $2$-from $F_{\#}^I$ on $X$ with integral periods.
In section~\ref{sec:sigmaT} we interpreted this form as a curvature of
a connection $\Theta_{\#}$ on a principal torus bundle $p:Y\to X$.
To perform the $T$-duality one has to construct this torus bundle and connection on it explicitly.
The torus bundle is defined by a $2$-cocycle on $X$. From equation \eqref{fabc}
it follows that the information contained in $\{h_{\alpha\beta}^I\}$
should be used to construct such a cocycle.
Moreover
locally $H|_{X_{\alpha}}=dB_{\alpha}$ so it is natural to ask whether $B$ \textit{alone}
defines the connection $\Theta_{\#}$. From equation
\begin{equation*}
F_{\#}^I\bigr|_{X_{\alpha}}=\imath(\tfrac{\pd}{\pd\theta_I})\,dB_{\alpha}
=\Lc(\tfrac{\pd}{\pd\theta_I})B_{\alpha}-d\,\imath(\tfrac{\pd}{\pd\theta_I})B_{\alpha}
\end{equation*}
it follows that the necessary condition for this is the invariance of $B_{\alpha}$
under the torus action: $\Lc(\tfrac{\pd}{\pd\theta_I})B_{\alpha}=0$ in all patches $X_{\alpha}$.
In particular, this condition implies that $H_0^{IJK}=0$.
If $\Lc(\tfrac{\pd}{\pd\theta_I})B_{\alpha}\ne0$ in some of the patches, then
one has to introduce   extra structure  into the formulation; to simplify the discussion, we will restrict ourselves to the case in which $B$ is invariant here.\footnote{See  Appendix \ref{app:t-gerbe} for an outline of a discussion of a  more general torus action on the gerbe connection $B$.}

The invariance of $B_{\alpha}$ with respect to the torus action restricts the
possible dependence of $\{h_{\alpha\beta}^I\}$,
 $\{a_{\alpha\beta}\}$ and $\{f_{\alpha\beta\gamma}\}$ on the torus coordinates:
the right hand sides of \eqref{BBgluing} must be pullbacks from the base.
The result can be summarized by the following

\begin{thm}
The gluing conditions for the gerbe connection $B_{\alpha}$ which are compatible with
the $\mathbb{T}^n$-invariance $\Lc(\frac{\pd}{\pd\theta_I})B_{\alpha}=0$ are
\begin{subequations}
\begin{align}
B_{0\alpha}^{IJ}-B_{0\beta}^{IJ}&=m^{IJ}_{\alpha\beta},
\\
B_{1\alpha}^I-B_{1\beta}^I&=d_M\tilde{h}_{\alpha\beta}^I
+m_{\alpha\beta}^{IJ}\bigl(A_{\beta}-\tfrac12\,d_M\lambda_{\beta\alpha}\bigr)_J,
\\
B_{2\alpha}-B_{2\beta}&=
[d_M\tilde{a}_{\alpha\beta}+\langle\tilde{h}_{\alpha\beta},F\rangle]
+\frac12\,
\bigl\langle m_{\alpha\beta},(A_{\beta}-\tfrac12\,d_M\lambda_{\beta\alpha})\wedge
(A_{\beta}-\tfrac12\,d_M\lambda_{\beta\alpha})\bigr\rangle
\end{align}
\label{BBBgluing}
\end{subequations}
where $\{m_{\alpha\beta}^{IJ}\}$ are skewsymmetric integral valued matrices satisfying
the cocycle condition on threefold overlaps, $\{\tilde{h}_{\alpha\beta}^I\}$
are functions (skew-symmetric in $\alpha,\beta$) defined on twofold overlaps $\{M_{\alpha\beta}\}$ and satisfying
the following condition on threefold overlaps
\begin{equation}
m_{\alpha\beta}+m_{\beta\gamma}+m_{\gamma\alpha}=0
\quad\text{and}\quad
\tilde{h}_{\alpha\beta}^I+\tilde{h}_{\beta\gamma}^I+\tilde{h}_{\gamma\alpha}^I\bigl|_{M_{\alpha\beta\gamma}}
=-\frac{1}{2}\,\bigl[m_{\alpha\beta}^{IJ}\lambda_{\beta\gamma\,J}
-m_{\gamma\beta}^{IJ}\lambda_{\beta\alpha\,J}\bigr].
\end{equation}
$\{\tilde{a}_{\alpha\beta}\}$ are $1$-forms defined on twofold overlaps and satisfying
the following condition on threefold overlaps:
\begin{multline}
\tilde{a}_{\alpha\beta}+\tilde{a}_{\beta\gamma}+
\tilde{a}_{\gamma\alpha}\bigl|_{M_{\alpha\beta\gamma}}
=\frac{1}{2\pi i}\,d_M\log \tilde{f}_{\alpha\beta\gamma}
-\frac{1}{12}\,d_M\bigl[\lambda_{\beta\alpha}(m_{\alpha\beta}+m_{\gamma\beta})\lambda_{\beta\gamma}\bigr]
\\
-\frac18\bigl(
\lambda_{\beta\alpha}m_{\gamma\beta}d_M\lambda_{\beta\alpha}
+\lambda_{\beta\alpha}m_{\gamma\alpha}d_M\lambda_{\beta\gamma}
+\lambda_{\beta\gamma}m_{\beta\alpha}d_M\lambda_{\beta\gamma}
+\lambda_{\beta\gamma}m_{\gamma\alpha}d_M\lambda_{\beta\alpha}
\bigr)
\end{multline}
where $\tilde{f}_{\alpha\beta\gamma}:M_{\alpha\beta\gamma}\to U(1)$ and it satisfies the
following condition on fourfold overlaps:
\begin{equation}
\tilde{f}_{\alpha\beta\gamma}\tilde{f}_{\beta\gamma\delta}^{-1}
\tilde{f}_{\gamma\delta\alpha}\tilde{f}_{\delta\alpha\beta}^{-1}=
\exp\left[
-\frac{2\pi i}{6}
\bigl(
\lambda_{\delta\gamma}m_{\delta\beta}\lambda_{\alpha\delta}
-\lambda_{\beta\gamma}m_{\alpha\delta}\lambda_{\delta\gamma}
+\lambda_{\beta\delta}m_{\gamma\delta}\lambda_{\delta\alpha}
\bigr)
\right].
\label{4fold}
\end{equation}
\end{thm}

\vspace{3mm}

Before proving the theorem let us discuss the implications of the result.
The invariance of the gerbe connection with respect to the torus action,  $\Lc(\tfrac{\pd}{\pd\theta_I})B_{\alpha}=0$, while not being the most general case,
allows for  gluing functions that are sufficiently nontrivial. In particular, $H_1^{IJ}$ can represent
a nontrivial de Rham cohomology class. The corresponding integral cohomology class
\begin{wrapfigure}{l}{160pt}
%\vspace{-5mm}
\includegraphics[width=140pt]{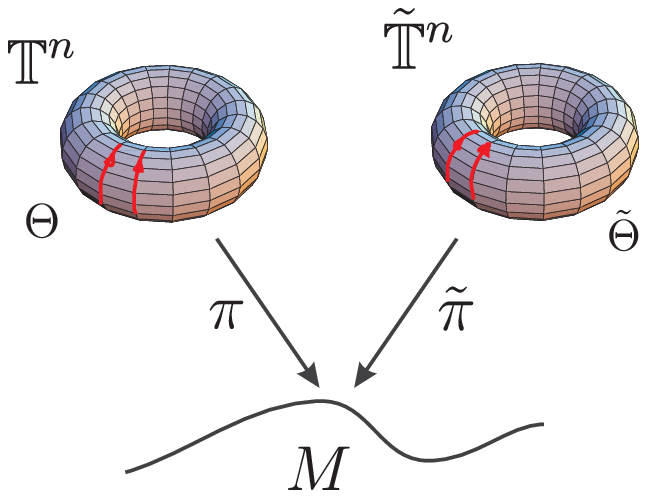}
\end{wrapfigure}
is represented by a cocycle $\{m_{\alpha\beta}\}$.
If $\{m_{\alpha\beta}\}$ is a coboundary (so it can be set to zero) then
$B_0$ is a globally well defined smooth function, $B_{1\alpha}$ has gluing
functions $\{\tilde{h}_{\alpha\beta}^I\}$ corresponding to a connection on a principal torus bundle:

\begin{cor}
\label{cor:0}
If $\{m_{\alpha\beta}\}$ is a coboundary then $Y$ is a principal double torus fibration.
The gluing functions are $\lambda_{\alpha\beta}$ and $\tilde{h}_{\alpha\beta}$:
\begin{equation}
\theta_{\alpha}-\theta_{\beta}=-\lambda_{\alpha\beta}
\quad\text{and}\quad
\tilde{\theta}_{\alpha}-\tilde{\theta}_{\beta}=-\tilde{h}_{\alpha\beta}.
\end{equation}
The connection one forms are $\Theta_I$ and $\tilde{\Theta}^I=d\tilde{\theta}^I_{\alpha}+B_{1\alpha}^I$.
$Y$ can also be thought of as a fibrewise product of two principal torus bundles, $X$ and $\tilde X$, defined by the gluing functions $\lambda_{\alpha\beta}$ and $\tilde{h}_{\alpha\beta}$ respectively.
\end{cor}

In general, $\{m_{\alpha\beta}\}$ is a nontrivial cocycle and $Y$ is a double fibration ---
a principal torus bundle $Y\to X$ over a principal torus bundle $X\to M$:
\begin{cor}
\label{cor:1}
The following functions defined on twofold overlaps $\{X_{\alpha\beta}\}$
\begin{equation}
\lambda_{\#\alpha\beta}(\theta_{\beta})=\tilde{h}_{\alpha\beta}
-m_{\alpha\beta}(\theta_{\beta}+\tfrac12\,\lambda_{\beta\alpha})
\end{equation}
satisfy the cocycle condition on threefold overlaps $\{X_{\alpha\beta\gamma}\}$.
This cocycle defines a principal torus bundle $p:Y\to X$ by the gluing condition
$\theta_{\#\alpha}-\theta_{\#\beta}=-\lambda_{\#\alpha\beta}$. The connection
$\Theta_{\#}$ on $Y$ can locally be written as
\begin{equation}
\Theta_{\#}^I\bigr|_{X_{\alpha}}=d\theta_{\#\alpha}^I+B_{1\alpha}^I-B_{0\alpha}^{IJ}\Theta_J.
\end{equation}
\end{cor}

The same space $Y$ can be represented in a slightly different geometrical form:
$Y$ is an affine double torus bundle (for an introduction to affine torus bundles see section 4.1.1 in \cite{pdb})
with very special gluing functions:
\begin{cor}
\label{cor:2}
The gerbe reduction \eqref{BBBgluing} defines an affine $\mathbb{T}^{n}\times\mathbb{T}^n$-torus bundle over $M$:
the gluing conditions for coordinates on twofold overlaps are
\begin{equation}
\begin{pmatrix}
\tilde{\theta}_{\alpha}+\frac12\,\tilde{h}_{\alpha\beta}
\\
\theta_{\alpha}+\frac12\lambda_{\alpha\beta}
\end{pmatrix}
=
\begin{pmatrix}
\mathbbmss{1} & m_{\alpha\beta}
\\
0 & \mathbbmss{1}
\end{pmatrix}
\begin{pmatrix}
\tilde{\theta}_{\beta}+\frac12\,\tilde{h}_{\beta\alpha}
\\
\theta_{\beta}+\frac12\lambda_{\beta\alpha}
\end{pmatrix}.
\end{equation}
The corresponding affine connection has the form
$\Theta_{\alpha}=d\theta_{\alpha}+A_{\alpha}$ and
$\tilde{\Theta}_{\alpha}=d\tilde{\theta}_{\alpha}+B_{1\alpha}$
with the gluing condition on twofold overlaps
\begin{equation}
\begin{pmatrix}
\tilde{\Theta}_{\alpha}
\\
\Theta_{\alpha}
\end{pmatrix}
=
\begin{pmatrix}
\mathbbmss{1} & m_{\alpha\beta}
\\
0 & \mathbbmss{1}
\end{pmatrix}
\begin{pmatrix}
\tilde{\Theta}_{\beta}
\\
\Theta_{\beta}
\end{pmatrix}.
\end{equation}
It also follows that $\Theta_{\#}^I|_{Y_{\alpha}}=\tilde{\Theta}_{\alpha}^I-B_{0\alpha}^{IJ}\,\Theta_J$ is globally well
defined $1$-form on the total space of the affine torus bundle.\footnote{In fact, similar affinisation happens when one considers a simpler case of  reduction of a $U(1)$ bundle on a principal torus bundle $X$. The basic steps are the same, but the derivation is much lighter and is presented in Appendix \ref{app:u1}.}
\end{cor}

\vspace{5mm}
\begin{proof}[Proof of the theorem:]
From the invariance of $B_{\alpha}$
on the torus coordinates it follows that the right hand sides of equations
\eqref{BBgluing} do not depend on the torus coordinates. On the other hand
the curvature $F_{\alpha\beta}$ of the connection $A_{\alpha\beta}$
\begin{equation*}
F_{\alpha\beta}=\dots+\frac12\,m_{\alpha\beta}^{IJ}
\,\Theta_I\wedge\Theta_J
\quad\text{where}\quad
m_{\alpha\beta}^{IJ}=\frac{\pd}{\pd\theta_I}h^J_{\alpha\beta}
-\frac{\pd}{\pd\theta_J}h^I_{\alpha\beta}
\end{equation*}
must have integral periods. In particular, $m_{\alpha\beta}^{IJ}$ is an integral valued
matrix. So locally (it means one has to cover the torus $\mathbb{T}^n$ by patches)
we can write $h_{\alpha\beta}^I(\theta_{\beta})$ as
\begin{equation*}
h_{\alpha\beta}^J(\theta_{\beta})=\pi^*\tilde{h}_{\alpha\beta}^J
+\frac12\,m_{\alpha\beta}^{IJ}\bigl(\theta_{\beta}+\tfrac12\lambda_{\beta\alpha}\bigr)_{\!I}.
\end{equation*}
Note that $h_{\alpha\beta}(\theta_{\beta})=-h_{\beta\alpha}(\theta_{\alpha})$ provided
$\tilde{h}_{\alpha\beta}$ is skewsymmetric in $\alpha,\beta$.
From the gluing condition for $B_{1\alpha}$ \eqref{BBgluing} we
know that
\begin{equation*}
\pi^*B_{1\alpha}^I-\pi^*B_{1\beta}^I=\pi^*(d_M\tilde{h}_{\alpha\beta}^I)
+\frac12\,m_{\alpha\beta}^{JI}\bigl(\tfrac12\lambda_{\beta\alpha}-A_{\beta}\bigr)_{\!J}
-\Lc(\tfrac{\pd}{\pd\theta_I})\,a_{\alpha\beta}
\end{equation*}
should not depend on the torus coordinates. This means that $a_{\alpha\beta}$
is at most linear a function of the torus coordinates:
\begin{equation*}
a_{\alpha\beta}(\theta_{\beta})=\pi^*\tilde{a}_{\alpha\beta}
+\langle\pi^*\rho_{\alpha\beta},\theta_{\beta}+\tfrac12\lambda_{\beta\alpha}\rangle
\end{equation*}
where $\tilde{a}_{\alpha\beta}$ and $\rho_{\alpha\beta}$ are smooth $1$-forms on $M_{\alpha\beta}$.
Again $a_{\alpha\beta}(\theta_{\beta})=-a_{\beta\alpha}(\theta_{\alpha})$ provided
$\tilde{a}_{\alpha\beta}$ and $\rho_{\alpha\beta}$ are skewsymmetric in $\alpha,\beta$.
From the gluing condition for $B_{2\alpha}$:
\begin{equation*}
B_{2\alpha}-B_{2\beta}=\pi^*[d\tilde{a}_{\alpha\beta}+\langle\tilde{h}_{\alpha\beta},F\rangle]
+\bigl(d\rho_{\alpha\beta}^I+\tfrac12 m_{\alpha\beta}^{IJ}F_J\bigr)(\theta_{\beta}+\tfrac12\lambda_{\beta\alpha})_I
+\langle\rho_{\alpha\beta},A_{\beta}-\tfrac12 d\lambda_{\beta\alpha}\rangle.
\end{equation*}
one concludes that
\begin{equation*}
\rho_{\alpha\beta}^I=-\frac12\, m_{\alpha\beta}^{IJ}\bigl(A_{\beta}
-\tfrac12 d\lambda_{\beta\alpha}\bigr)_{\! J}\,.
\end{equation*}
Combining the equations above we obtain the gluing conditions for the gerbe connection \eqref{BBBgluing}.

To obtain the cocycle conditions one has to study equations \eqref{fabc}. From the first
equation in \eqref{fabc} we learn that
\begin{equation*}
f_{\alpha\beta\gamma}(\theta_{\beta})=
\exp\left[-\frac{2\pi i}{4}\bigl(\theta_{\beta}+\tfrac13\lambda_{\beta\alpha}
+\tfrac13\lambda_{\beta\gamma}\bigr)\bigl(m_{\alpha\beta}\lambda_{\beta\gamma}
-m_{\gamma\beta}\lambda_{\beta\alpha}\bigr)\right]
\,\pi^*\tilde{f}_{\alpha\beta\gamma}
\end{equation*}
where $\tilde{f}_{\alpha\beta\gamma}$ satisfies the usual symmetric properties: $\tilde{f}_{\beta\alpha\gamma}=
\tilde{f}_{\alpha\beta\gamma}^{-1}$ etc.
Straightforward calculation yields the relation \eqref{4fold} on fourfold overlaps.
\end{proof}

\comment{Note that the connection $A_{\alpha\beta}$ \eqref{Adecomp}
for the invariant gerbe connection has the form
\begin{equation}
A_{\alpha\beta}(\theta_{\beta})=
\pi^*\tilde{a}_{\alpha\beta}+\langle \pi^*\tilde{h}_{\alpha\beta},\Theta\rangle
+\frac12\,m_{\alpha\beta}^{IJ}\bigl(\theta_{\beta}+\tfrac12\lambda_{\beta\alpha}\bigr)_{\!I}\,
\bigl(d\theta_{\beta}+\tfrac12 d\lambda_{\beta\alpha}\bigr)_{\!J}.
\end{equation}
}

%%%%%%%%%%%%%%%%%%%%%%%%%%%%%%%%%%%%%%%%%%%%%%%%%%%%
\section{T-duality in string sigma models}
\label{sec:stdTdual}\setcounter{equation}{0}
In this section we discuss $T$-duality for principal
torus bundles with nontrivial $H$-flux (with vanishing $H_0^{IJK}$).
First, we consider the case when $B_0^{IJ}$ is a smooth function on $M$
and present the standard derivation of the $T$-duality
on the level of function integral.
Second, we discuss the problems with the generalization for the case when $B_0$
is not globally defined.

\subsection{Sigma model on a principal torus bundle}
\label{sec:Tdualsigma}
In this section we review the construction of the sigma model with
a target space   $X$ that is a principal torus bundle $\pi:X\to M$.
The space $\mathrm{Map}(\Sigma,X)$ of maps from $\Sigma$ to $X$
has itself a structure of a fibre bundle\footnote{Strickly speaking, it is
not a conventional fibre bundle but rather   one that is defined for Fr\'echet manifolds (see e.g.  \cite{Felder:1988sd, Gawedzki:1987ak}).}:
\begin{equation}
\Gamma(x^*X)\hookrightarrow
\mathrm{Map}(\Sigma,X)\to \mathrm{Map}(\Sigma,M).
\label{MapF}
\end{equation}
This fibre bundle is defined as follows:
given a map $\chi \in \mathrm{Map}(\Sigma,X)$ we define $x=\pi\circ \chi$ as
the composition of this map followed by the projection onto the base manifold $M$. This defines
the map $x:\Sigma\to M$. Now we can restrict the principal torus bundle $X\to M$ to the
image of $\Sigma$ and then pull it back to $\Sigma$. The fibre in \eqref{MapF}
is exactly the space of sections of the resulting torus bundle.
 It is convenient to write the sigma model
functional integral in the following factorized form
\begin{subequations}
\begin{align}
\mathcal{Z}(\mathfrak{g})&=\int_{\mathrm{Map}(\Sigma,M)} \hspace{-1.2cm}\Ds x^{\mu}(\sigma)\,\exp\left[-\pi\int_{\Sigma}g_{\mu\nu}(x)\,dx^{\mu}\wedge *_{\mathfrak{g}}dx^{\nu}
+\frac{1}{4\pi}\int_{\Sigma}\vol(\mathfrak{g})\mathcal{R}(\mathfrak{g})\,\Phi(x)\right]
\,\Psi(x(\sigma));
\\
\Psi(x)&=\int_{\Gamma(x^*X)}
\hspace{-0.8cm}
\Ds \theta_I(\sigma)\,
\exp\left[
-\pi\int_{\Sigma}h^{IJ}(x)\,\Theta_I\wedge *_{\mathfrak{g}}\Theta_J\right]\,
\notag
\\
&\hspace{6cm}\times\mathrm{Hol}\bigl(x^*B_2+\langle x^*B_1,\Theta\rangle
+\tfrac12\langle x^*B_0,\Theta\wedge\Theta\rangle, \Sigma\bigr).
\label{Phi(x)}
\end{align}
\end{subequations}
Here $\Psi(x(\sigma))$ is a function of the map $x$ to the base space $M$,
$g_{\mu\nu}$ is a metric on the base $M$, $h_{IJ}(x)$ is an invariant metric
on the torus fibre over   $x\in M$, $\mathfrak{g}$ is a metric
on the worldsheet $\Sigma$, $\Phi(x)$ is the dilaton. $\mathrm{Hol}(B,\Sigma)$
denotes the holonomy of the gerbe connection $B$ on $\Sigma$ --- the exponential of the Wess-Zumino
term, see Appendix~\ref{app:WZ}.
Note that the dilaton is only a function
of the base coordinates. In string theory, one also has to calculate the integral over the space
of 2d metrics, however we are not going to discuss this integral here.

\subsection{$T$-duality}
In this section we review the derivation of $T$-duality in a sigma model on
a Riemann  surface $\Sigma$ with
  target space  a principal torus bundle $\pi:X\to M$.

The main result of this subsection can be summarized by the following
\begin{thm}
\label{thm:31}
If $B_0^{IJ}$ is globally well defined then
the functional \eqref{Phi(x)} can also be written as
\begin{multline}
\Psi(x(\sigma))=\left[\frac{\det h^{IJ}(x)}{\det\tilde{h}_{IJ}(x)}\right]^{\chi(\Sigma)/2}
\int_{\Gamma(x^*{\tilde X})}\hspace{-0.8cm}\Ds \tilde{\theta}^I(\sigma)\,
\exp\left[
-\pi\int_{\Sigma}\tilde{h}_{IJ}(x)\,\tilde{\Theta}^I\wedge *_{\mathfrak{g}}\tilde{\Theta}^J \right]
%\right.
\\
%\left.
\times\mathrm{Hol}\bigl( x^*\tilde{B}_2-\langle \tilde{\Theta},x^*A\rangle
+\frac12\langle\tilde{\Theta}\wedge\tilde{\Theta},\tilde{B}_0\rangle , \Sigma \bigr)
\label{anotherPhi(x)}
\end{multline}
where $\tilde X$ is defined in Corollary \ref{cor:0}, $\chi(\Sigma)$ is the Euler character of the surface $\Sigma$,
\begin{equation*}
\tilde{h}=(h-B_0h^{-1}B_0)^{-1}\quad\text{and}\quad
\tilde{B}_0=-\tilde{h}B_0h^{-1}
\end{equation*}
are
symmetric and antisymmetric parts of the matrix $(h+B_0)^{-1}$ respectively, and
\begin{equation*}
\tilde{\Theta}^I=d\tilde{\theta}^I+x^*B_1\quad\text{and}\quad
\tilde{B}_2=B_2+\langle B_1,A\rangle.
\end{equation*}
\end{thm}

\begin{cor}
Under classical $T$-duality the set $(F_I,H_3,\tilde{F}^I)$
maps to the set $(\tilde{F}^I,H_3,F_I)$ where $\tilde{F}=H_2+B_0F$
is the curvature of the connection $\tilde{\Theta}$,
and
\begin{equation*}
\Phi(x)\mapsto \Phi(x)+\frac12\log\left[\frac{\det h^{IJ}(x)}{\det\tilde{h}_{IJ}(x)}\right].
\end{equation*}
\end{cor}

\begin{proof}
If $X$ were a product space then one could prove the theorem in the standard way:
gauge the torus symmetry, add lagrange multipliers to impose the condition that the gauge-fields be pure gauge, change the order of integration
and integrate out the original torus variables \cite{Giveon:1991jj,Rocek:1991ps}.
However
when the equivariant extension of $H$ does not exist, and equations (\ref{equiv}) are not satisfied
this approach does not work: \textit{it is impossible to gauge the sigma model}.
It was shown in \cite{Hull:2006qs} that although it is impossible to gauge the sigma
model it makes sense to gauge and add the lagrange multiplier in one step.
This result is explained in detail in Lemma~\ref{lemma:1}.
The extended sigma model is defined by the functional integral:
\begin{multline}
\Psi(x(\sigma))=\int \Ds \theta_I\Ds \Lambda_I\Ds\tilde{\theta}^I\,
\exp\left[
-\pi\int_{\Sigma}h^{IJ}(x)(\Theta_I-\Lambda_I)\wedge *(\Theta_J-\Lambda_J)
\right]\times
\\
\text{``}\exp\text{''}\left[
2\pi i\int_{\Sigma}\bigl(x^*B_2+\langle x^*B_1,\Theta-\Lambda\rangle
+\frac12\langle x^*B_0,(\Theta-\Lambda)\wedge(\Theta-\Lambda)\rangle+\langle d\tilde{\theta},d\theta-\Lambda\rangle\bigr)
\right].
\label{proof1}
\end{multline}
where $\Lambda_I$ is a globally well defined smooth $1$-form on $\Sigma$,\footnote{In \cite{Hull:2006qs}, the combination $x^*(A - \Lambda)$ is denoted   $C$.}  and $\tilde{\theta}^I(\sigma)$
is a section of the pullback of the principal torus bundle with   fibre $\tilde{\mathbb{T}}^n$
as described in section~\ref{sec:sigma}. The gluing conditions on the twofold overlaps
for $\tilde{\theta}$ are exactly those described in Corollary~\ref{cor:0}.
The exponential in \eqref{proof1} is invariant with respect to the gauge transformations
$\theta(\sigma)\mapsto \theta(\sigma)+\phi(\sigma)$ and $\Lambda\mapsto \Lambda+d_{\Sigma}\phi$.
We can rewrite the last term in \eqref{proof1} in a slightly different way:
\begin{multline*}
\Psi(x(\sigma))=\int \Ds \theta_I\Ds \Lambda_I\Ds\tilde{\theta}^I\,
\exp\left[
-\pi\int_{\Sigma}h^{IJ}(x)(\Theta_I-\Lambda_I)\wedge *(\Theta_J-\Lambda_J)
\right]
\\
\times\mathrm{Hol}(B_2+\langle B_1,A\rangle
-\langle \tilde{\Theta},A\rangle,\Sigma)
\exp\left[2\pi i\int_{\Sigma}\bigl(\langle \tilde{\Theta},\Theta-\Lambda\rangle
+\frac12\langle B_0,(\Theta-\Lambda)\wedge(\Theta-\Lambda)\rangle\bigr)
\right].
\end{multline*}
Note that the last line can be rewritten as an integral over a $3$-disk $D$ with   boundary $\Sigma$:
\begin{equation}
e^{2\pi i\int_{\Sigma}(\dots)}
=e^{2\pi i\int_{D}(x^*H_3+\langle x^*H_2,\Theta-\Lambda\rangle+
\frac12\langle H_1,(\Theta-\Lambda)\wedge(\Theta-\Lambda)\rangle+\langle
\tilde{\Theta}-B_0(\Theta-\Lambda),d\Lambda\rangle)}.
\end{equation}
Note that the exponential on the right hand side contains only globally well defined quantities.

Imposing a gauge-fixing condition on
  $\theta_I$ and integrating over $\Lambda$ yields \eqref{anotherPhi(x)} (see Lecture~7 of E.~Witten in \cite{qft} for details).
\end{proof}

\comment{
\begin{enumerate}
\item
Note that the component $B_2$ of the gerb connection \eqref{B} is
not invariant under $T$-duality transformations (see also \cite{Fidanza:2003zi}). It   transforms
as in Theorem~\ref{thm:31}.

\item Theorem~\ref{thm:31} can be generalized to cover the action of
the whole $T$-duality group $O(n,n;\Zh)$. The pair $(F_I,\tilde{F}^I)$
transforms as a vector of $O(n,n;\Zh)$ while the matrix $h^{IJ}+B_0^{IJ}$ transforms
by fractional linear transformations.
\end{enumerate}
}

\begin{lemma}
\label{lemma:1}
Assuming $B_0$ is globally well defined the
functional $\Psi(x)$ in \eqref{proof1} as a function
of the gerbe connection $B$ descends to a well defined function
of the gauge equivalent classes of gerbe connections
(or in short $\Psi(x)$ is gauge invariant).
\end{lemma}

\begin{proof}
Suppose that the image of $\Sigma$, $x(\Sigma)$, lies in the patch $M_{\alpha}$
then we can write \eqref{proof1} in two two different ways: using coordinates
in the patch $M_{\alpha}$ or in the patch $M_{\beta}$.
First notice that the Jacobian in change of measure between the patches $M_{\alpha}$
and $M_{\beta}$ is trivial, so
\begin{equation*}
\Ds x\Ds\theta_{\alpha}\Ds \tilde{\theta}_{\alpha}|_{M_{\alpha\beta}}
=\Ds x\Ds\theta_{\beta}\Ds \tilde{\theta}_{\beta}|_{M_{\alpha\beta}}.
\end{equation*}
The only nontrivial term we can obtain is from the second line in \eqref{proof1}.
So let us rewrite the second line of \eqref{proof1} written in patch $M_{\alpha}$ in terms
of the quantities defined in the patch $M_{\beta}$:
\begin{multline}
\left.\exp\left[2\pi i\int_{\Sigma}(x^*B_{2\alpha}+\langle x^*B_{1\alpha},\Theta-\Lambda\rangle
+\frac12\langle x^*B_{0\alpha},(\Theta-\Lambda)\wedge(\Theta-\Lambda)\rangle
+\langle d\tilde{\theta}_{\alpha},d\theta_{\alpha}-\Lambda\rangle)
\right]\right|_{M_{\alpha\beta}}
\\
=
e^{2\pi i\int_{\Sigma}(\dots)_{\beta}}\bigr|_{M_{\alpha\beta}}
\left.\exp\left[2\pi i\int_{\Sigma}\langle \tilde{\Theta}_{\beta}-\tilde{\Theta}_{\alpha},\Lambda\rangle\right]\right|_{M_{\alpha\beta}}
\times
\\
\exp\left[2\pi i\int_{\Sigma}d_{\Sigma}\,x^*(
\tilde{a}_{\alpha\beta}+\langle\tilde{h}_{\alpha\beta},\Theta\rangle
+\langle d\tilde{\theta}_{\beta}+\frac12d\tilde{h}_{\beta\alpha},\lambda_{\alpha\beta}\rangle
-\langle\tilde{h}_{\alpha\beta},d\theta_{\beta}+\frac12d\lambda_{\beta\alpha}\rangle
)\right]
\end{multline}
where $\tilde{\Theta}_{\alpha}=d\tilde{\theta}_{\alpha}+B_{1\alpha}$.
To cancel the second exponent one has to require that $\tilde{\Theta}$ is a globally
well defined $1$-form which means that $\tilde{\theta}_{\alpha}$ is a coordinate
on a principal torus bundle as in Corollary~\ref{cor:1}. The third exponent vanishes
by itself since it is an integral of a total derivative over the compact closed surface $\Sigma$.
\end{proof}

\subsection{$[H_1]_{dR}\ne 0$}

In this subsection we rederive the result  that it is impossible to construct the gauged sigma model
with extended target space when $H_1$ is not exact  \cite{Hull:1989jk, Hull:1990ms}.

The simplest way to see this is to
consider the 3-dimensional form of  the WZ term $\int x^*H$, integrated
over a  3-space whose boundary is the world-sheet,
The first step is minimal coupling, i.e.
the replacement  $x^*\Theta$ to $x^*\Theta-\Lambda$ where
$\Lambda$ is a globally defined $1$-from on $\Sigma$. The lagrangian
should be a closed 3-form.
 The minimally-coupled $3$-form
\begin{subequations}
\begin{equation}
x^*H_3+\langle x^*H_2,\Theta-\Lambda\rangle+\frac12\langle x^*H_1,(\Theta-\Lambda)^2\rangle
\end{equation}
is not closed. To make it closed we add two terms: one proportional to $d\Lambda$ and
another proportional to $\Lambda\wedge d\Lambda$,
\begin{equation}
\langle G+w\,\Lambda,d\Lambda\rangle
\end{equation}
\label{gaugedH}
\end{subequations}
where $G$ is a $1$-form on $X$ and $w^{IJ}$ is a function on $X$.
The closure of \eqref{gaugedH} yields:
\begin{equation}
dG^I=F^I_{\#},\quad w^{IJ}=-w^{JI},\quad H_1^{IJ}=dw^{IJ},\quad
\Lc(\tfrac{\pd}{\pd\theta_I})\,w=0.
\end{equation}
 The invariance with
respect to the shift $\theta\to \theta+\phi(\sigma)$ requires $G^I=G_1^I-w^{IJ} \Theta_J$
where $G_1$ is a pullback of  $1$-form from $M$. To make \eqref{gaugedH} globally defined
requires $w$ to be globally defined. So we conclude that $H_1$ is an exact form.
From the discussion in section~2 it follows that one can take $w=B_0$
and $G=\Theta_{\#}$.

If one continues the discussion of lemma~\ref{lemma:1} one obtains
that $\Psi(x)$ is not gauge invariant any more. It descends to a
section of a non-trivial line bundle over the space of gauge equivalence
classes of gerbe connections.

%%%%%%%%%%%%%%%%%%%%%%%%%%%%%%%%%%%%%%%%%%%%%%%%%%%%
\section{T-duality as a symmetry of a loop space}
\label{sec:loop}\setcounter{equation}{0}
In this section we discuss $T$-duality in terms of the  canonical quantization of the phase space
of the sigma model.
For an earlier treatment in which $T$-duality is understood as a canonical transformation, see
\cite{Alvarez:1994wj}.
A bosonic string sigma model on $S^1\times \Rh$ with the target space $X$
has the configuration space $LX$ --- the loop space of $X$ -- and the phase space $T^*LX$.
We show that when $B_0=0$ (or more generally when $B_0$ is globally well defined),
$T$-duality is a symmetry of a total space
of a line bundle over the cotangent bundle to the loop space on $Y$. The symplectic form $\omega_Y$ has two different torus actions and the two corresponding hamiltonian reductions yield the two $T$-dual models.

When  $B_0 \neq 0$ and is topologically nontrivial, there is still a  symmetry
but it  is realized differently. There is a Hamiltonian action of one torus but the other has a non Hamiltonian action. The obstruction to having a Hamiltonian action is that $[H_1]_{dR}\ne 0$ (i.e. there can only be a Hamiltonian action if $[H_1]_{dR}= 0$).

In the following subsections we will review the construction of the sigma model phase space for a general
smooth manifold $X$ and then restrict to the case in which  $X$ is a principal torus bundle.

\subsection{Phase space of string sigma model}
Let $X$ be a compact smooth manifold. The phase space for the string sigma model on $S^1\times \Rh$
is $T^*LX$ --- the cotangent bundle to the loop space of $X$. $T^*LX$ is naturally
a symplectic space: given a loop $x:S^1\hookrightarrow X$, the symplectic form is
\begin{equation}
\omega=\oint_{S^1}d\sigma\,\delta p
=\oint_{S^1}d\sigma\,\delta p_{M}(\sigma)\wedge\delta x^{M}(\sigma).
\label{omega}
\end{equation}
One can think of the momentum $p=p_{M}(\sigma)\delta x^{M}(\sigma)$ as of a section of the
pullback of $T^*X$ to $S^1$. Here $\delta$ is the differential on the loop space and $x^M$ are coordinates on $X$.
To quantize the theory   we follow the standard procedure of geometrical quantization (see e.g. \cite{Axelrod:1989xt})  and
specify a hermitian line bundle over the phase space with a connection  that has curvature $\omega$.
Since $\omega$ is exact one can take a trivial line bundle and choose a connection
\begin{equation}
\vartheta=\delta z+\oint_{S^1}d\sigma\,p
\end{equation}
where $z\in\Ch$ is a coordinate on the fibre.
 The wave-functions are then sections of this bundle.

We are interested in   sigma models which are twisted by a $B$-field.
Mathematically, the $B$-field is a gerbe connection, and the relevance of this
will become clear shortly. The gerbe connection has a curvature $H$ --- a closed globally defined
smooth $3$-form on $X$ with integral periods. Using $H$ we can twist
the symplectic form \eqref{omegaX} to give
\begin{equation}
\omega_X=\oint_{S^1}d\sigma\,\bigl[\delta p
+\imath(\pd_{\sigma}x)H\bigr].
\label{omegaX}
\end{equation}
Here
\begin{equation}
\oint_{S^1}d\sigma\,\imath(\pd_{\sigma}x)H=\frac 12 \oint_{S^1}d\sigma\, \pd_{\sigma}x ^M(\sigma)H_{MNP}(x(\sigma)) \delta x^{N}(\sigma)\wedge \delta x^{P}(\sigma)
\end{equation}
is a $2$-form on loop-space.

To quantize this phase space we specify a hermitian line bundle
over $T^*LX$
with connection $\vartheta_X$ whose curvature is $\omega_X$ { (The space of $L^2$ sections
of this line bundle form a prequantum Hilbert space)}. Now the magic fact is
that a gerbe connection on $X$ defines a principal circle bundle over $LX$ with connection
whose curvature is exactly the second term in \eqref{omegaX} \cite{Hitchin:2005uu}.
Then the   pullback of this circle bundle to $T^*LX$ can be taken as the required line bundle.

\paragraph{Connection and circle action.}
The connection on the line bundle over $T^*LX$ can be written as (in the patch $X_{\alpha}$)
\begin{equation}
\vartheta_X=\delta z_{\alpha} +
\oint_{S^1}d\sigma\,\bigl[p-\imath(\pd_{\sigma}x)B_{\alpha}\bigr].
\label{loop:connection}
\end{equation}
with $\imath(\pd_{\sigma}x)B=  \pd_{\sigma}x^MB_{MN}\delta x^{N}(\sigma)$.
Recall that $B$ is not a globally defined $2$-form, rather $B_{\alpha}-B_{\beta}=dA_{\alpha\beta}$
on the twofold intersection $X_{\alpha\beta}$.
The momentum $p$ is nevertheless globally well defined: it is a section of $x^*(T^*X)$.
So one sees that $z_{\alpha}-z_{\beta}=\oint_{S^1} x^*A_{\alpha\beta}$.

The associated circle action is given by the group of line bundles with connection:
a line bundle $L\to X$ with connection $A$ acts on the $B$-field by the shift $B_{\alpha}\mapsto B_{\alpha}
+F|_{X_{\alpha}}$ where $F$ is the curvature of the connection on $L$.
The pullback $x^*L$ of the line bundle $L$ to the loop $S^1$ is necessarily  a trivial line
bundle with a flat connection:  the
second Cheeger-Simons cohomology
(essentially the space of connections modulo gauge transformations; see e.g. \cite{Freed:2006yc}) is
$\check{H}^2(S^1)\cong U(1)$.
In other words, a pullback of the gerbe to a loop is a principal homogenous space for
$U(1)$, so we have a principal circle bundle.
From \eqref{loop:connection} it is easy to see
that\footnote{Note that
$
\imath(\pd_{\sigma}x)\delta A=\Lc(\pd_\sigma x)(x^*A)-\delta(\imath(\pd_{\sigma}x A))
=\pd_{\sigma}[\imath(\pd_{\sigma}x)A]-\delta(\imath(\pd_{\sigma}x)\, A).
$}
the coordinate $z$ shifts $z\mapsto z+\oint_{S^1}x^*A$.

\paragraph{Sigma model on a principal torus bundle.}
Let $X$ be the principal torus bundle $\pi:X\to M$ which we described in section~\ref{sec:sigmaT}.
In order to define a sigma model on $X$
we have to specify a gerbe connection. We use coordinates $x^\mu$ on the base $M$ and fibre coordinates $\theta_I$, as before, so that $x^M=(x^\mu,\theta_I)$.

A connection \eqref{loop:connection} for a target $X$ which is a principal   torus bundle has the following form
\begin{equation}
\vartheta_X=\delta z+\oint_{S^1}d\sigma\,\bigl[
p_{\mu}\delta x^{\mu}+\langle p,\Theta\rangle -\imath(\pd_{\sigma}x+\nabla_{\sigma}\theta)B
\bigr]
\label{varthetaX}
\end{equation}
where $x:S^1\to M$ defines a loop on the base manifold $M$,
and $\theta \in\Gamma(x^*X)$ is section of the pullback torus bundle;
 $\imath(\pd_{\sigma}x)$ stands for
$\imath(\pd_{\sigma}x^{\mu}(\frac{\pd}{\pd x^{\mu}})_{\texttt{hor}})$ and
\begin{equation}
\nabla_{\sigma}\theta_I=\pd_{\sigma}\theta_I+\imath(\pd_{\sigma}x)A_{I}
\end{equation}
is the covariant derivative of $\theta$ with respect to the pullback connection.
Explicitly,   $\langle p,\Theta\rangle =
p^I(\sigma)[\delta\theta_I(\sigma) + A_{I\mu}(x(\sigma))\,\delta x^{\mu}(\sigma)]$ etc.

\subsection{Symmetry of a loop space}

Recall that, under the assumptions of section~\ref{sec:sigmaT}, the
generalized correspondence space $Y$ is a double torus bundle over $M$ (provided $B_0^{IJ}$ is
globally well defined). We shall first  discuss the case when $B_0^{IJ}=0$.

The cotangent bundle to the loop space of $Y$ is naturally a symplectic
manifold with $\omega_Y=\delta \vartheta_Y$. Here $\vartheta_Y$ is a connection
on the corresponding line bundle:
\begin{equation}
\vartheta_Y=\delta z_{\alpha}+\oint_{S^1}d\sigma\,\Bigl[
p_{\mu}\delta x^{\mu}+\langle p,\Theta\rangle+\langle{\tilde \Theta} ,{\tilde p}\rangle
-\imath(\pd_{\sigma}x+\nabla_{\sigma}\theta)\bigl(B_{\alpha}-\langle{\tilde \Theta},\Theta\rangle\bigr)
\Bigr].
\label{varthetaY}
\end{equation}
Note that there are two extra terms in this connection
  than were in  \eqref{varthetaX}: the meaning of the first one is obvious,
 while the second can be interpreted as the topologically trivial
gerbe connection comming from the Poincar\' e line bundle (See e.g. \cite{Bouwknegt:2003vb} for an  explanation of this and other relevant geometric structures.). The sympletic form is
\begin{multline}
\omega_Y=\oint_{S^1}d\sigma\,\Bigl[
\delta p_{\mu}\wedge \delta x^{\mu}+\langle \delta p,\Theta\rangle
+\langle p,F\rangle
+\langle H_2,{\tilde p}\rangle-\langle {\tilde \Theta},\delta {\tilde p}\rangle
\\
+\imath(\pd_{\sigma}x+\nabla_{\sigma}\theta)\bigl(H_3+\langle {\tilde \Theta},F\rangle\bigr)
\Bigr]
\label{omegaY}
\end{multline}
where we have used   $d{\tilde \Theta}^I=p^*H_2 = {\tilde \pi}^* {\tilde F}$ (and we are assuming $B_0^{IJ}=0$).
The main result of this subsection is the following:
\begin{thm}
The symplectic form $\omega_Y$ is invariant with respect to the $\mathbb{\tilde T}^n \times
\mathbb{T}^n$ action generated by the fundamental vector fields $\frac{\delta}{\delta {\tilde \theta}^I}$
and $\frac{\delta}{\delta\theta_I}$. Moreover this action is hamiltonian:
\begin{subequations}
\begin{align}
\imath(\tfrac{\delta}{\delta{\tilde \theta}^I})\omega_Y&=\delta(- {\tilde p}_I+\nabla_{\sigma}\theta_I);
\\
\imath(\tfrac{\delta}{\delta\theta_I})\omega_Y&=\delta(-p^I+\nabla_{\sigma}{\tilde \theta}^I).
\end{align}
\end{subequations}
The symplectic reduction with respect to ${ \mathbb{\tilde T}}^n$ or $\mathbb{T}^n$
yields the symplectic forms $\omega_X$ or $\omega_{\tilde X}$ respectively where
\begin{subequations} \label{thr-tilde}
\begin{align}
\omega_X&=\oint_{S^1}d\sigma\,\Bigl[
\delta p_{\mu}\wedge\delta x^{\mu}
+\langle\delta p,\Theta\rangle+\langle p,F\rangle
+\imath(\pd_{\sigma}x+\nabla_{\sigma}\theta)\bigl(H_3+\langle {\tilde F},\Theta\rangle\bigr)
\Bigr];
\\
\omega_{\tilde X}&=\oint_{S^1}d\sigma\,\Bigl[
\delta p_{\mu}\wedge\delta x^{\mu}
-\langle{\tilde \Theta},\delta  {\tilde p} \rangle+\langle  {\tilde F}, {\tilde p} \rangle
+\imath(\pd_{\sigma}x+\nabla_{\sigma}{\tilde \theta})\bigl(
H_3+\langle{\tilde \Theta},F\rangle\bigr)
\Bigr].
\end{align}
\end{subequations}
\end{thm}

\begin{cor}
Instead of doing symplectic reduction with respect to ${ \mathbb{\tilde T}}^n$ or $\mathbb{T}^n$
one can reduce with respect to some sub-torus inside $\mathbb{\tilde T}^n \times \mathbb{T}^n$.
The space of such sub-tori is an affine space with the group
of translations being given by $O(n,n;\Zh)$. This space encompasses all $T$-dual backgrounds.
\end{cor}

Clearly, the symmetry   in (\ref{thr-tilde}) corresponds to the $T$-duality  exchange that was discussed in section \ref{sec:stdTdual} (for $B_0^{IJ}=0$).

\subsection{Non-hamiltonian torus action}

%In this subsection we investigate what happens with the original torus action if $B_0^{IJ}$ is nontrivial.
As was discussed in section~\ref{sec:sigma},
if $B_0^{IJ}$ is non-zero the original torus action can in general
be only lifted to an  $\Rh^n$ action.  However we have also seen that if $B_0^{IJ}$ is globally
well-defined then  topologically $Y$ is a principal $2n$-torus bundle over $M$, and
there exists another lift which defines the torus action.  In other words, $B_0^{IJ} \neq 0$ is a geometrical obstruction for $Y$ being a principal torus bundle with connection over $M$, but when it is well defined  there exists a connection (\ref{tTh}),
$\tilde{\Theta}^I=\Theta_{\#}^I-B_0^{IJ}\Theta_J$, which respects the double-torus fibered structure of $Y$. We may extend the construction of the previous subsection to the case when $B_0^{IJ} \neq 0$ and is not necessarily globally defined.

\paragraph{Symplectic form on Y.}
A connection on a line bundle over the cotangent bundle to the loop space of $Y$
can be written in a form similar to \eqref{varthetaY}:
\begin{equation}
\vartheta_Y=\delta z+\oint_{S^1}d\sigma\,\Bigl[
p_{\mu}\delta x^{\mu}+\langle p,\Theta\rangle+\langle\Theta_{\#},p^{\#}\rangle
-\imath(\pd_{\sigma}x+\nabla_{\sigma}\theta)\bigl(B-\langle\Theta_{\#},\Theta\rangle\bigr)
\Bigr]
\label{varthetaYb}
\end{equation}
where $\nabla_{\sigma}\theta_I$ is as before but
\begin{equation}
\Theta^I_{\#}=\delta\theta^I_{\#}+B_1^I-B_0^{IJ}\Theta_J \, .
\end{equation}
Note that $\vartheta_Y$ is written in terms of globally well defined connection $\Theta_{\#}$.
One easily sees that if $B_0=0$ the equation \eqref{varthetaYb} reduces to
\eqref{varthetaY}.
The sympletic form is
\begin{multline}
\omega_Y=\oint_{S^1}d\sigma\,\Bigl[
\delta p_{\mu}\wedge \delta x^{\mu}+\langle \delta p,\Theta\rangle
+\langle p,F\rangle
+\langle F_{\#},p^{\#}\rangle-\langle \Theta_{\#},\delta p^{\#}\rangle
\\
+\imath(\pd_{\sigma}x+\nabla_{\sigma}\theta)H-\delta\langle\Theta_{\#},\nabla_{\sigma}\theta\rangle
\Bigr]
\label{omegaYB}
\end{multline}
where we have used   $d\Theta_{\#}^I=p^*F_{\#}^I$. Recall that $F_{\#}^I = \imath(\frac{\pd}{\pd\theta_I})H$
and thus it is globally well defined.

The action of $\mathbb{T}^n_{\#}$ on the symplectic form $\omega_Y$ is still hamiltonian:
\begin{equation}\label{red-ha}
\imath\bigl(\tfrac{\delta}{\delta\theta_{\#}^I}\bigr)\,\omega_Y=\delta(-p^{\#}_I+\nabla_{\sigma}\theta_I).
\end{equation}
So the hamiltonian reduction by $\mathbb{T}^n_{\#}$
yields a sigma model with the symplectic form $\omega_X = \delta \vartheta_X$ constructed from the connection (\ref{varthetaX}). The action of  $\mathbb{T}^n$ however is no longer hamiltonian:
\begin{equation}
\label{non-ha}
\imath\bigl(\tfrac{\delta}{\delta\theta_I}\bigr)_{\texttt{hor}} \omega_Y = \delta(-p^I+\nabla_{\sigma}\theta_{\#}^I) + H_1^{IJ} (-p_J^{\#}+\nabla_{\sigma}\theta_J)\, ,
\end{equation}
where $(\tfrac{\delta}{\delta\theta_I})_{\texttt{hor}}|_{M_{\alpha}}=\tfrac{\delta}{\delta\theta_{I\alpha}}
-B_{0\alpha}^{IJ}\frac{\delta}{\delta\theta_{\#\,\alpha}^J}$  denotes the horizontal lift
of the vector field $\frac{\delta}{\delta\theta}$ via the connection $\Theta_{\#}$.
This is explained by the failure of the $\mathbb{T}^n$ action  to lift to $Y$ and it is not surprising that  $H_1^{IJ}$ defines the obstruction to the symplectic reduction.\footnote{Actually this is a weaker condition - having a constant irrational $B_0^{IJ}$ is sufficient for the failure of the torus action to lift. Once more, (\ref{non-ha}) holds regardless  whether $B_0^{IJ}$ is globally defined or not.}

When  $B_0^{IJ}$ is well defined this situation can be fixed:
the vector fields $\frac{\pd}{\pd\theta_I}$ and $\frac{\pd}{\pd\theta^I_{\#}}$
are independently well defined and thus we can rewrite equation \eqref{non-ha} as
\begin{multline}
\label{ha-ha}
\imath\bigl(\tfrac{\delta}{\delta\theta_I}\bigr) \omega_Y
 -B_0^{IJ}\imath\bigl(\tfrac{\delta}{\delta\theta^J_{\#}}\bigr) \omega_Y
 = \delta(-p^I+\nabla_{\sigma}\theta_{\#}^I) + \delta B_0^{IJ} (-p_J^{\#}+\nabla_{\sigma}\theta_J)
 \\
\Rightarrow\quad\imath\bigl(\tfrac{\delta}{\delta\theta_I}\bigr)\omega_Y
= \delta\bigl(-p^I-B_0^{IJ}p_J^{\#}+\underbrace{\nabla_{\sigma}{ \theta_{\#}}^I + B_0^{IJ} \nabla_{\sigma}\theta_J}_{\nabla_{\sigma}\tilde{\theta}^I} \bigr) .
\end{multline}
Now on substituting $\Theta_{\#}^I=\tilde{\Theta}^I + B_0^{IJ}\Theta_J$ and redefining $(p^I- B^{IJ}_0 p^{\#}_J, p_I^{\#}) \mapsto (p^I , \tilde{p}_I)$ one finds that $\omega_Y$
can be written as
\begin{multline}
\omega_Y=\oint_{S^1}d\sigma\,\Bigl[
\delta p_{\mu}\wedge \delta x^{\mu}+\langle \delta p,\Theta\rangle
+\langle p,F\rangle
+\langle {\tilde F} ,{\tilde p}\rangle-\langle {\tilde \Theta},\delta {\tilde p}\rangle
\\
+\imath(\pd_{\sigma}x+\nabla_{\sigma}\theta)\bigl(H_3+\langle {\tilde \Theta},F\rangle\bigr)
\Bigr]
\label{omegaY-tilde}
\end{multline}
One can check that there are now two hamiltonian torus actions with respective reductions yielding symplectic forms $\omega_X$ and $\omega_{\tilde X}$ related via $(\Theta_I, p^I) \leftrightarrow ({\tilde \Theta}^I, {\tilde p}_I)$.

When $B_0^{IJ}$ is not globally well defined,  similar steps can be made but instead of
(\ref{omegaY-tilde}) one obtains:

\begin{thm}\label{thm:last}
The symplectic form $\omega_Y$ on $T^*LY$ can be written as
\begin{multline}
\omega_Y=\oint_{S^1}d\sigma\,\Bigl[
\delta p_{\mu}\wedge \delta x^{\mu}+\langle \delta p_{\alpha},\Theta\rangle
+\langle p_{\alpha},F\rangle
+\langle \tilde{F}_{\alpha} ,{\tilde p}\rangle
-\langle \tilde{\Theta}_{\alpha},\delta \tilde{p}\rangle
\\
+\imath(\pd_{\sigma}x+\nabla_{\sigma}\theta)\bigl(H_3+\langle \tilde{\Theta}_{\alpha},F\rangle\bigr)
\Bigr]
\label{omegaY-new}
\end{multline}
where the both the momenta $\tilde{p}_{\alpha},\,p_{\alpha}$ and connections $\tilde{\Theta}_{\alpha},\,
\Theta_{\alpha}$ are not globally defined. The gluing functions on twofold overlaps $M_{\alpha\beta}$
are
\begin{equation}
\begin{pmatrix}
\tilde{p}_{\alpha}
\\
p_{\alpha}
\end{pmatrix}
=
\begin{pmatrix}
\mathbbmss{1} & 0
\\
m_{\alpha\beta} & \mathbbmss{1}
\end{pmatrix}
\begin{pmatrix}
\tilde{p}_{\beta}
\\
p_{\beta}
\end{pmatrix}
\quad\text{and}\quad
\begin{pmatrix}
\tilde{\Theta}_{\alpha}
\\
\Theta_{\alpha}
\end{pmatrix}
=
\begin{pmatrix}
\mathbbmss{1} & m_{\alpha\beta}
\\
0 & \mathbbmss{1}
\end{pmatrix}
\begin{pmatrix}
\tilde{\Theta}_{\beta}
\\
\Theta_{\beta}
\end{pmatrix}.
\end{equation}
The expressions \eqref{omegaYB} and \eqref{omegaY-new}  are the same.
Note that although most of the terms in \eqref{omegaY-new} are not well defined --- their sum
is well defined, and thus can be integrated.
\end{thm}

 Thus the string sigma model can be consistently quantized  in the canonical approach with a phase space   constructed from  the
generalized correspondence space $Y$, which    an affine doubled torus
fibration over a base manifold $M$ even in the case in which the
$B_0^{IJ}$ component of $B$ is not globally well defined. This supports
the
view that passing to $Y$ is the correct way of dealing with sigma models
in situations in which the gauging and $T$-duality is obstructed. The
phase space on $T^*LY$ with symplectic form on $\omega_Y$ has a natural
$O(n,n)$ action and puts momentum and winding modes on an equal footing,
so that gluing functions mixing the two can be incorporated easily. It
seems that the phase space can then be defined in situations in which
there is no well defined dual configuration space.

\section*{Acknowledgments}
We would like to thank P.~Bouwknegt, M. Gra\~na, A.~ Lawrence, M.~Petrini, R.~Szabo, S.~Theisen, A.~Tseytlin and D.~Waldram for discussions.
We would like to acknowledge the hospitality of the
Albert Einstein Institute, Potsdam (D.B. and R.M.), Isaac Newton Institute, Cambridge (D.B. and C.H.) and Yukawa Institute, Kyoto (R.M.) during the course of this work.  D.B. was supported by PPARC and in part by RFBR grant 05-01-00758; R.M. is supported in part by RTN contracts  MRTN-CT-2004-005104 and  MRTN-CT-2004-512194 and by ANR grant BLAN06-3-137168. R.M. would like to thank the string theory group and the Institute for Mathematical Sciences at Imperial College for warm hospitality, and  EPSRC for support.

%%%%%%%%%%%%%%%%%%%%%%%%%%%%%%%%%%%%%%%%%%%%%%%
\appendix
\section{Wess-Zumino term and holonomy of the gerbe connection}
\label{app:WZ}\setcounter{equation}{0}
In this appendix we review a definition of Wess-Zumino term or logarithm
of a gerbe connection for topologically nontrivial $B$-field.

\paragraph{Holonomy of an abelian $1$-form gauge field.}
Given a contractible cover $\{M_{\alpha}\}$ of a manifold $M$
a $1$-form gauge field is specified by the following data
\begin{itemize}
\item a function  $\lambda_{\alpha\beta}$ on each twofold intersection $M_{\alpha\beta}$
 satisfying the cocycle condition that
 $\lambda_{\alpha\beta}+\lambda_{\beta\gamma}+\lambda_{\gamma\alpha}=0$ on threefold overlaps.

\item a $1$-form $A_{\alpha}$  on each  $M_{\alpha}$  such
that on twofold overlaps $M_{\alpha\beta}$:
 $A_{\alpha}-A_{\beta}|_{M_{\alpha\beta}}=d\lambda_{\alpha\beta}$.
\end{itemize}

A loop $\gamma:S^1\to M$ does not necessarily lies within one patch.
We break the loop into segments $\{\gamma_{\alpha}\subset  M_{\alpha}\}$ and
denote by $\gamma_{\alpha\beta}\in M_{\alpha\beta}$ a point where the
segments $\gamma_{\alpha}$ and $\gamma_{\beta}$ intersect. Then the
logarithm of the holonomy of the gauge field $A$ is defined as the following sum
\begin{equation}
\frac{1}{2\pi i}\log\mathrm{Hol}(A,\gamma)=\sum_{\{\gamma_{\alpha}\}}\int_{\gamma_{\alpha}}A_{\alpha}
+\sum_{\gamma_{\alpha\beta}}\lambda_{\alpha\beta}(\gamma_{\alpha\beta}).
\end{equation}
One can easily verify that this sum does not depend on a particular choice of
the partitioning $\{\gamma_{\alpha},\gamma_{\alpha\beta}\}$ of the loop $\gamma$.

\paragraph{Holonomy of a gerbe connection.}
A gerbe connection is defined by a set of $2$-forms $B_{\alpha}$, $1$-forms $\{A_{\alpha\beta}\}$
on twofold intersections and functions $f_{\alpha\beta\gamma}:M_{\alpha\beta\gamma}\to U(1)$
on threefold intersections (see section~2.2 for details). Given a $2$-cycle $\Sigma$
we can partition it as shown in the picture: $\Sigma_{\alpha}$ are
\begin{wrapfigure}{l}{100pt}
\includegraphics[width=90pt]{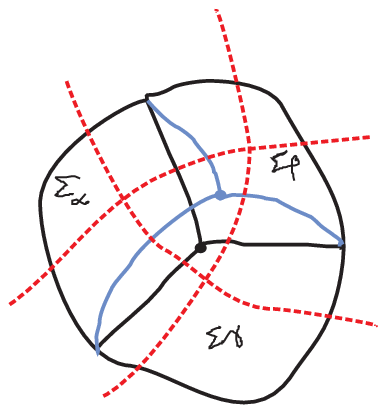}
\end{wrapfigure}
surfaces, $\Sigma_{\alpha\beta}$
is the common boundary of $\Sigma_{\alpha}$ and $\Sigma_{\beta}$ and $\Sigma_{\alpha\beta\gamma}$
is the intersection of segments $\Sigma_{\alpha\beta}$, $\Sigma_{\beta\gamma}$
and $\Sigma_{\gamma\alpha}$. The holonomy of the gerbe connection is defined
by the following sum (the orientation is important)
\begin{equation}
\frac{1}{2\pi i}\log \mathrm{Hol}(B,\Sigma)
=\sum_{\{\Sigma_{\alpha}\}}\int_{\Sigma_{\alpha}}B_{\alpha}
+\sum_{\{\Sigma_{\alpha\beta}\}}\int_{\Sigma_{\alpha\beta}}A_{\alpha\beta}
+\sum_{\{\Sigma_{\alpha\beta\gamma}\}}\frac{1}{2\pi i}\log
f_{\alpha\beta\gamma}(\Sigma_{\alpha\beta\gamma}).
\end{equation}
It requires a little bit more work to verify that this sum does not
depend on a particular choice of partitioning of the $2$-cycle $\Sigma$.

%%%%%%%%%%%%%%%%%%%%%%%%%%%%%%%%%%%%%%%%%%%%%
\section{Torus action on the gerbe connection}
\label{app:t-gerbe}
To discuss a group action on a sigma model one has to specify how it acts
on the  target space, metric and on any additional structure involved.
In this paper we have assumed that the space $X$ is a principal torus bundle ($\mathbb{T}^n$
acts freely on it) and that the metric is $\mathbb{T}^n$-invariant. We now discuss
how the torus group acts on a  gerbe connection with $\mathbb{T}^n$-invariant curvature.

In this appendix we choose \textit{a contractible} covering $\{\mathcal{U}_{\alpha}\}$ of the target space $X$. The main result can be summarized by the following
\begin{thm}
Let $\{\mathcal{U}_{\alpha}\}$ be a contractible covering of the space $X$. The
action of the torus group on a gerbe connection  with $\mathbb{T}^n$-invariant curvature is specified by
a $1$-forms $\{w_{\alpha}^I\}$ in every path $\mathcal{U}_{\alpha}$,
a function $u_{\alpha\beta}^I$ in each twofold overlap $\mathcal{U}_{\alpha\beta}$
and a constant $c_{\alpha\beta\gamma}$ is each threefold overlap $\mathcal{U}_{\alpha\beta\gamma}$ such
that they satisfy the following conditions
\begin{subequations}
\begin{align}
\Lc(\tfrac{\pd}{\pd\theta_I})B_{\alpha}&=dw_{\alpha}^I;
\label{a}
\\
(w_{\alpha}^I-w_{\beta}^I)\bigr|_{\mathcal{U}_{\alpha\beta}}
&=\Lc(\tfrac{\pd}{\pd\theta_I})A_{\alpha\beta}+du_{\alpha\beta}^I;
\label{b}
\\
(u_{\alpha\beta}^I+u_{\beta\gamma}^I+u_{\gamma\alpha}^I)\bigr|_{\mathcal{U}_{\alpha\beta\gamma}}
&=c_{\alpha\beta\gamma}^I
-\frac{1}{2\pi i}\,\Lc(\tfrac{\pd}{\pd\theta_I})\log f_{\alpha\beta\gamma};
\label{c}
\\
(c_{\alpha\beta\gamma}^I-c_{\beta\gamma\delta}^I+c_{\gamma\delta\alpha}^I
-c_{\delta\alpha\beta}^I)\bigr|_{\mathcal{U}_{\alpha\beta\gamma\delta}}&=m_{\alpha\beta\gamma\delta}^I\in\Zh.
\label{d}
\end{align}
\end{subequations}
\label{thm:B}
\end{thm}
\begin{proof}
The invariance of the curvature
$H$ of the gerbe connection implies $\Lc(\frac{\pd}{\pd\theta_I})H=0$. From this   it follows
that $\Lc(\frac{\pd}{\pd\theta_I})B_{\alpha}$ is  a closed form. Since the patch $\mathcal{U}_{\alpha}$
is contractible this closed form is exact, so we denote it by $dw_{\alpha}$. From
the gluing conditions for $B_{\alpha}$ we obtain \eqref{b} for some $u_{\alpha\beta}^I$, from the gluing
condition for $A_{\alpha\beta}$ one obtains \eqref{c}. Finally equation \eqref{d} comes
from the cocycle condition on $f_{\alpha\beta\gamma}$.
\end{proof}

Note that THE simplest   solution to  these equations is that in which $B_{\alpha}$ is invariant
with respect to the torus action in each patch.

Given a structure specified in the Theorem~\eqref{thm:B} we can construct a connection
$\Theta_{\#}$ and check whether or not we have a principal torus bundle.
 The curvature $F_{\#}=d \Theta_{\#}$ is
\begin{equation}
F_{\#}^I|_{\mathcal{U}_{\alpha}}=d\{w_{\alpha}^I-\imath(\tfrac{\pd}{\pd\theta_I})B_{\alpha}\}.
\end{equation}
The gluing conditions for these $1$-forms  are
\begin{equation}
\{w_{\alpha}^I-\imath(\tfrac{\pd}{\pd\theta_I})B_{\alpha}\}
-\{w_{\beta}^I-\imath(\tfrac{\pd}{\pd\theta_I})B_{\beta}\}
=d\{\imath(\tfrac{\pd}{\pd\theta_I})A_{\alpha\beta}+u_{\alpha\beta}^I\}.
\end{equation}
On threefold overlap the functions defined on the right hand side satisfy
\begin{equation}
\{\imath(\tfrac{\pd}{\pd\theta_I})A_{\alpha\beta}+u_{\alpha\beta}^I\}\bigr|_{\mathcal{U}_{\alpha\beta\gamma}}
+\dots=c_{\alpha\beta\gamma}.
\end{equation}
Thus $\{c_{\alpha\beta\gamma}\}$ is an obstruction to the geometrization of  $F_{\#}^I$:
if it does not vanish then one does not have a principal torus bundle with connection over $X$
whose curvature is $F_{\#}^I$.

%%%%%%%%%%%%%%%%%%%%%%%%%%%%%%%%%%%%%%%%%%%%%
\section{Reduction of a $U(1)$ bundle}
\label{app:u1}\setcounter{equation}{0}

We consider here a toy example of a reduction of a principal circle bundle
$L\stackrel{p}{\longrightarrow} X$ onto the torus fibration
$\mathbb{T}^n\hookrightarrow X\stackrel{\pi}{\longrightarrow}M$, specified
in subsection \ref{sec:sigmaT}. This example captures the essential
features of the affine bundle appearing in the gerbe reduction given in
section \ref{sec:gerbe}, but is considerably simpler.

We denote by $\Theta_{\#}$ a connection $1$-form on the total space
$L$ of the circle bundle. Locally it can be written as
\begin{equation}
\Theta_{\#}|_{X_{\alpha}}=d\psi_{\alpha}+p^*B_{\alpha}
\end{equation}
where $\psi_{\alpha}$ ($0\leqslant \psi_{\alpha\,I}<1$) is a coordinate on the circle in the patch $X_{\alpha}$ and $\{B_{\alpha}\}$  is a $1$-form.
We denote by $H\in\Omega^2_{\Zh}(X)$ the curvature of this connection, and assume
that both $H$ and $\{B_{\alpha}\}$ are invariant with respect to the torus action:
\begin{equation*}
\Lc\bigl(\tfrac{\pd}{\pd\theta_I}\bigr)H = 0
 \quad\text{and}\quad
\Lc\bigl(\tfrac{\pd}{\pd\theta_I}\bigr)B_{\alpha}=0.
\end{equation*}
The gluing conditions on two-fold overlaps $\{X_{\alpha\beta}\}$ are
\begin{equation}
\psi_{\alpha}|_{X_{\alpha\beta}}-\psi_{\beta}|_{X_{\alpha\beta}}=-\sigma_{\alpha\beta} \, , \qquad
B_{\alpha}\bigr|_{X_{\alpha\beta}}-B_{\beta}\bigr|_{X_{\alpha\beta}}=d\sigma_{\alpha\beta}
\label{psigluing}
\end{equation}
where $\{\sigma_{\alpha\beta}\}$ are functions on twofold overlaps satisfying the cocycle condition on threefold overlaps $X_{\alpha \beta \gamma}$:  $\sigma_{\alpha\beta}
+\sigma_{\beta\gamma}+\sigma_{\gamma\alpha}=0$.

The invariant $2$-form $H$ can be decomposed in horizontal forms:
\begin{equation}
H=\pi^*H_2+\langle \pi^*H_1,\Theta\rangle+\frac12\langle\pi^*H_0,\Theta\wedge\Theta\rangle,
\end{equation}
where $H_j$, $j=2,1,0$, are $j$-forms on the base manifold $M$. The assumption
of the invariance of the $1$-forms $B_{\alpha}$ with respect to the torus action
yields $H_0^{IJ}=0$.

In each coordinate patch we can decompose the $1$-form $B_{\alpha}$
into vertical and horizontal forms:
\begin{equation}
B_{\alpha}=B_{1\alpha}+\langle B_{0\alpha},\Theta\rangle,
\label{B1}
\end{equation}
where $B_{1\alpha}$ and $B_{0\alpha}$ are horizontal $1$- and $0$-forms respectively.
The gluing conditions take the form:
\begin{equation}
B_{1\alpha}\bigl|_{M_{\alpha\beta}}-B_{1\beta}\bigl|_{M_{\alpha\beta}}=(\pi^*d_M)\, \sigma_{\alpha\beta}\, , \qquad
B_{0\alpha}^I\bigl|_{M_{\alpha\beta}}-B_{0\beta}^I\bigl|_{M_{\alpha\beta}}= \Lc(\tfrac{\pd}{\pd\theta_I})\, \sigma_{\alpha\beta}.
\label{app:BB1gluing}
\end{equation}

The invariance of $\{B_{\alpha}\}$ with respect to the torus action yields
restrictions on a possible dependence of the gluing functions on the torus coordinates:
the right hand side of \eqref{app:BB1gluing} must be a pullback from the base manifold.
The most general solution of this conditions is
\begin{equation}
\sigma_{\alpha\beta}(\theta_{\beta})=\pi^*\tilde{\sigma}_{\alpha\beta}
+ \langle m_{\alpha\beta},\,\theta_{\beta}+\tfrac12\lambda_{\beta\alpha}\rangle ,
\label{sigma}
\end{equation}
where $m_{\alpha\beta}^{I}$ is an integral valued vector
and $\tilde{\sigma}_{\alpha\beta}$ is a function on the twofold overlap $M_{\alpha\beta}$
of the base manifold.
On the threefold overlaps  $\{M_{\alpha\beta\gamma} \}$ they satisfy the following conditions:
\begin{subequations}
\begin{align}
m_{\alpha\beta}^I + m_{\beta \gamma}^I + m_{\gamma \alpha}^I &=0 \, , \\
\tilde{\sigma}_{\alpha\beta} + \tilde{\sigma}_{\beta \gamma}+ \tilde{\sigma}_{\gamma \alpha} &= \frac{1}{2} \bigl(\langle m_{\alpha\beta},  \lambda_{\beta \gamma}\rangle
- \langle m_{\gamma \beta},  \lambda_{\beta \alpha}\rangle\bigr)
\end{align}
\label{newco}
\end{subequations}
where $\{\lambda_{\alpha\beta\,I}\}$ are gluing functions of the principal torus bundle $X$ (see
section~\ref{sec:sigma}).
The gluing conditions for $B_1$ and $B_0$ become
\begin{subequations}
\begin{align}
B_{1\alpha}\bigl|_{M_{\alpha\beta}}-B_{1\beta}\bigl|_{M_{\alpha\beta}}&=
d_M\tilde{\sigma}_{\alpha\beta} + \langle m_{\alpha\beta}, A_{\beta} - \tfrac12d_M\lambda_{\beta\alpha}\rangle \, ,
\\
B_{0\alpha}^I\bigl|_{M_{\alpha\beta}}-B_{0\beta}^I\bigl|_{M_{\alpha\beta}} &= m_{\alpha\beta}^{I}
\end{align}
\label{BB11gluing}
\end{subequations}
It is now not hard to see that the equation (\ref{sigma}) and the first of (\ref{BB11gluing}) define an affine
 $S^1\times\mathbb{T}^n$-torus bundle over $M$. On twofold overlaps ${M_{\alpha\beta}}$,
the gluing conditions for coordinates $\psi_{\alpha}$ and $\theta_{\alpha}$  and the affine connection
$\Theta_{\alpha}=d\theta_{\alpha}+A_{\alpha}$ and
$\tilde{\Theta}_{\alpha}=d\psi_{\alpha}+B_{1\alpha}$ are given by:
\begin{equation}
\begin{pmatrix}
\psi_{\alpha}+\frac12\,{\tilde \sigma}_{\alpha\beta}
\\
\theta_{\alpha}+\frac12\lambda_{\alpha\beta}
\end{pmatrix}
=
\begin{pmatrix}
1 & m_{\alpha\beta}
\\
0 & \mathbbmss{1}
\end{pmatrix}
\begin{pmatrix}
\psi_{\beta}+\frac12\,{\tilde \sigma}_{\beta\alpha}
\\
\theta_{\beta}+\frac12\lambda_{\beta\alpha}
\end{pmatrix} ,
\qquad
\begin{pmatrix}
\tilde{\Theta}_{\alpha}
\\
\Theta_{\alpha}
\end{pmatrix}
=
\begin{pmatrix}
1 & m_{\alpha\beta}
\\
0 & \mathbbmss{1}
\end{pmatrix}
\begin{pmatrix}
\tilde{\Theta}_{\beta}
\\
\Theta_{\beta}
\end{pmatrix}.
\end{equation}

Only when the gluing function $\sigma_{\alpha\beta}$ does not depend on torus coordinates, i.e.  $m_{\alpha\beta}^{I}  =0$, the reduction of the $U(1)$ bundle yields a $1$-form on $M$ and $n$ scalar fields. When the gluing function does not respect the torus action, the result of the reduction is given by an affine  $S^1\times\mathbb{T}^n$ fibration over $M$ and $n$ line bundles.

%%%%%%%%%%%%%%%%%%%%%%%%%%%%%%%%%%%%%%%%%%%%%
\section{Reduction of the current algebra}
\label{app:reductionCourant}\setcounter{equation}{0}

\paragraph{Current algebra.} Given a section $(v,\rho)$ of $TX\oplus T^*X$ one
can construct a current
\begin{equation}
J_{\epsilon}{(v,\rho)}=\oint_{S^1}d\sigma\,\epsilon(\sigma)
\bigl[\imath(v)p+\imath(\pd_{\sigma}x)\,\rho\bigr]
\end{equation}
where $\epsilon(\sigma)$ is a smooth (test) function on the circle.
From \eqref{omega} it follows that the Poisson bracket of two such currents is
\cite{Alekseev:2004np,Zabzine:2006uz,Guttenberg:2006zi}
\begin{equation}
\{J_{\epsilon_1}(v_1,\rho_1),\,J_{\epsilon_2}(v_2,\rho_2)\}
=J_{\epsilon_1\epsilon_2}\bigl([(v_1,\rho_1),(v_2,\rho_2)]_H\bigr)
-\frac12\oint_{S^1}d\sigma\,(\epsilon_1\pd_{\sigma}\epsilon_2-\epsilon_2\pd_{\sigma}\epsilon_1)
\bigl[\imath(v_1)\rho_2+\imath(v_2)\rho_1\bigr]
\label{JJ}
\end{equation}
where $[\cdot,\cdot]_H$ is the twisted Courant bracket. The twisted Courant
bracket is defined by
\begin{equation}
[(v_1,\rho_1),(v_2,\rho_2)]_H=[v_1,v_2]
+\Bigl\{\Lc({v_1})\rho_2-\Lc({v_2})\rho_1
-\frac12\,d(\imath({v_1})\rho_2-\imath({v_2})\rho_1)
+\imath({v_1})\imath({v_2}) H
\Bigr\}
\end{equation}
where $[\cdot,\cdot]$ denotes the commutator of vector fields.
Note that one can rewrite the Poisson bracket above in a slightly
different form: as a twisted Courant bracket on the $T\oplus T^*$  bundle over  $X\times S^1$
(see equation (30) in \cite{Alekseev:2004np}).

\paragraph{Reduction of the Courant bracket.}

Taking  $X$ to be  a principal torus bundle,  we can study the reduction
of the twisted current algebra to the base $M$.

We start by decomposing the sections of  $TX\oplus T^*X$ into horizontal and vertical components.  Any vector $v$ and one-form $\rho$ can be written as
\begin{align}
v &= v_M + \langle K, f\rangle  \nonumber \\
\rho &= \rho_M + \langle \phi, \Theta   \rangle \nonumber
\end{align}
Demanding that both ${\cal L}_K v = 0$ and  ${\cal L}_K \rho = 0$, implies in particular  $f \in \Omega^0(M,\mathfrak{t})$ and $\phi \in \Omega^0(M,  \mathfrak{t}^*)$.
In other words, a $\mathbb{T}^n$-invariant section of $TX$ can be written as an element $(v_M, f) \in
TM \oplus \mathfrak{t}$, while a  $\mathbb{T}^n$-invariant section of $TX^*$ can be written as $(\rho_M, \phi) \in T^*M\oplus  \mathfrak{t}^*$. Given these elements, we can introduce some basic operations replacing the contractions,  Lie brackets and  Lie derivatives:
\begin{align}
\imath({(v_M,f)}) (\lambda_M, \omega) &= \imath({v_M}) \lambda_M + \langle \omega, f \rangle  \nonumber \\
d(\lambda_M, \omega) &= (d\lambda_M + \langle  \omega, F \rangle, -d \omega)  \nonumber \\
{\cal L}_{(v_M,f)} (\lambda_M, \omega) &= ({\cal L}_{v_M} \lambda_M + \langle \omega , \imath({v_M}) F + df    \rangle, \, {\cal L}_{v_M}  \omega)  \nonumber \\
[ (v_M,f), (w_M, g)] &= ([v_m, w_M] ,  \imath({v_M})  \imath({w_M}) F + {\cal L}_{v_M} g - {\cal L}_{w_M} f )
\nonumber
\end{align}
In this notation, a contraction of the element  $(v_M, f)$ with a  $p$-form in $\Omega^p_{\Zh}(X)$ can be thought of as a collection of  forms in  $\Omega^i_{\Zh}(M, \Lambda^{p-i-1}\mathfrak{t})$ for $i=0,..., p-1$. In particular, for
 $H\in \Omega^3_{\Zh}(X)$,
\begin{equation*}
\imath({(v_M,f)}) H = \Bigl( (\imath({v_M})  H_3 + \langle H_2, f \rangle) , (\imath({v_M})  H_2  - \langle  H_1 , f \rangle) , (\imath({v_M})  H_1 + \langle  H_0, f \rangle) \Bigr)
\end{equation*}
We can now write down the reduction of the twisted Courant algebra to the base $M$ in a compact form.

\begin{thm}
\label{thm:A1}
The space of $\mathbb{T}^n$-invariant sections of $TX\oplus T^*X$
is isomorphic to $\Gamma(TM\oplus T^*M\oplus \mathfrak{t}\oplus \mathfrak{t}^*)$.
The Courant bracket on $TX\oplus T^*X$ yields the following
bracket on $\mathbb{T}^n$-invariant sections
\begin{align}
[&(v_M,f ; \, \rho_M, \phi),  \,(w_M, g; \lambda_M, \omega)]_H=
\Bigl( [ (v_M,f), (w_M, g)] ;
\nonumber \\
&{\cal L}_{(v_M,f)} (\lambda_M, \omega) - {\cal L}_{(w_M,g)} (\rho_M, \phi)
+ \frac{1}{2} (\imath({(v_M,f)}) (\lambda_M, \omega) - \imath({(w_M,g)}) (\rho_M, \phi) )
+ \imath({(v_M,f)}) \imath({(w_M,g)}) H \Bigr) \, .\nonumber
\end{align}
\end{thm}

\vspace{2mm}
\noindent
The reduced Courant bracket in Theorem~\ref{thm:A1} can be cast as
\begin{align}
[(v_M,f ;  & \rho_M, \phi),   \,(w_M, g;   \lambda_M, \omega)]_H=
[(v_M;  \rho_M ),  \,(w_M; \lambda_M)]_{H_3} + \nonumber \\
& \Bigl(0,{\cal L}_{v_M} g - {\cal L}_{w_M} f ; \langle \omega, df\rangle -  \langle  \phi, dg \rangle  +\frac{1}{2}d( \langle \omega, f \rangle  -  \langle  \phi, g \rangle  ), {\cal L}_{v_M} \omega - {\cal L}_{w_M} \phi   \Bigr) + \nonumber \\
&\Bigl( 0,  \imath({v_M})  \imath({w_M}) F  ;   \langle  \omega ,  \imath({v_M}) F \rangle + \langle  \imath({v_M}) F_{\#} , g\rangle  -  \langle   \imath({w_M}) F_{\#} , f\rangle  - \langle  \phi , \imath({w_M}) F  \rangle ,   \imath({v_M})  \imath({w_M})   F_{\#}  \Bigr) - \nonumber \\
&\Bigl(  0,0; \langle H_1, [f,g]  \rangle,  \langle H_0, [f,g] \rangle \Bigr)  \, .\nonumber
\end{align}
The rhs of the first line is the Courant bracket on the base $M$  twisted by  a $3$-form $H_3$, which in general will not be cl osed. Adding the second  line  amounts to extending the bracket  to  $M \times \mathbb{T}^n$ (or $M \times \mathbb{R}^n$)  \cite{CourantRed}. On the third line we recover  $F_{\#}^I=H_2^I-H_1^{IJ}\wedge \Theta_J
+\frac12\,H_0^{IJK}\Theta_J\wedge\Theta_K$; in the absence of nontrivial $B_1$ and $B_0$ it
displays an explicit $O(n,n, \mathbb{Z})$ symmetry, which exchanges the terms containing $F_I$ and $H^I_2$,  reflecting the fact that there are two independent principal tori  on $M$ (and thus two choices to which of the two forms corresponds to the curvature of the connection, and which to twisting).
  It may appear strange that in the general case the natural $O(n,n, \mathbb{Z})$ action is on $2$-forms $F_{\#}^I$ rather than ${\tilde F}^I$ preferred by the sigma model. However nontrivial $B^{IJ}_0$ and $H^{IJK}_0$ have  contributions that spoil this symmetry of the Courant bracket.\footnote{This fact has been noticed  in particular by P.~Bouwknegt.} These are collected in the last line (with $[.,.]$ denoting antisymmetrization in $I,J$ indices).

%%%%%%%%%%%%%%%%%%%%%%%%%%%%%%%%%%%%%%%%%%%%%%%
%%%%%%%%%%%%%%%%%%%%%%%%%%%%%%%%%%%%%%%%%%%%%%%

\end{document}